\documentclass[a4paper,11pt]{article}
\usepackage[pdftex]{color,graphicx}
\usepackage{pdfpages}

\pdfoutput=1

\usepackage{theorem}
\usepackage{amsmath}
\usepackage{amssymb}
\usepackage{proof}
\usepackage{graphics}

{\theorembodyfont{\rmfamily}
  \newtheorem{definition}{Definition}[section]
  \newtheorem{lemma}[definition]{Lemma}
  \newtheorem{theorem}[definition]{Theorem}
  \newtheorem{corollary}[definition]{Corollary}
  \newtheorem{remark}[definition]{Remark}

  \newtheorem{property}[definition]{Property}
  
  \newtheorem{example}[definition]{Example}
}

\def\pushright#1{{              % set up
    \parfillskip=0pt            % so \par doesnt push \square to left
    \widowpenalty=1000000000         % so we dont break the page before \square
    \displaywidowpenalty=10000  % ditto
    \finalhyphendemerits=0      % TeXbook exercise 14.32
    \leavevmode                 % \nobreak means lines not pages
    \unskip                     % remove previous space or glue
    \nobreak                    % don't break lines
    \hfil                       % ragged right if we spill over
    \penalty50                  % discouragement to do so
    \hskip.2em                  % ensure some space
    \null                       % anchor following \hfill
    \hfill                      % push \square to right
    {#1}                        % the end-of-proof mark (or whatever)
    \par}}                      % build paragraph

\newenvironment{proof}{{\bf Proof: }}
{\pushright{\rule{.5em}{.5em}}\penalty-700 \smallskip}

\def\ra{\rightarrow}
\def\Ra{\Rightarrow}
\def\la{\leftarrow}

\def\da{\downarrow}
\def\Da{\Downarrow}
\def\ua{\uparrow}
\def\Ua{\Uparrow}

\def\ltimes{\otimes}

\def\lbd{\lambda}
\def\lam{\lambda}

\def\ul{\underbar}

\def\lws{\lambda_{ws}}
\def\lwsn{\lambda_{wsn}}

\def\ls{\lambda\sigma}
\def\lsigma{\lambda\sigma}
\def\lsn{\lambda\sigma_n}

\def\lbdch{\overline\lbd\mu\tilde\mu}

\def\es{\emptyset}
\def\sm{\setminus}

\def\bigspace{\;\;\;\;\;\;\;\;\;\;\;\;\;\;\;\;\;\;\;\;\;}

\newcommand{\lab}[1]{\langle #1 \rangle}

\newcommand{\Upshift}[3]{\mathcal{U}_{#1}^{#2}(#3)}
\newcommand{\Fliftcons}[2]{I_{#1}(#2)}
\newcommand{\Fliftshift}[2]{J_{#1}(#2)}
\newcommand{\Fshift}[2]{K_{#1}(#2)}
\newcommand{\ol}[1]{\overline{#1}}
\newcommand{\ub}[1]{\underline{#1}}

\def\skelin{\preccurlyeq}
\def\rel{\lessdot}
\def\relsim{\mbox{\rotatebox{90}{\!$\lessdot$}}}

\def\trg{\triangleright}

\def\espace{\;\;\;\;\;\;\;\;\;\;\;\;\;\;\;\;\;\;}

\newcommand{\indice}[1]{1\underset{#1}{\underbrace{[\ua]...[\ua]}}}

\title{Deriving SN from PSN: a general proof technique}
\author{Emmanuel Polonowski}
\date{}

\begin{document}

  \thispagestyle{empty}

\vspace*{-3.5cm}
\hspace*{-3.5cm}
\begin{minipage}{18cm}

  \noindent\includegraphics[height=2.1cm]{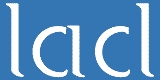}\hfill
  \includegraphics[height=2.1cm]{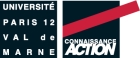}

  \vspace*{0.5cm}

  %\hrule width 20cm
  \hrule

  \vspace*{1cm}

  \begin{center}
    {\LARGE\bf Deriving SN from PSN:}
  \end{center}

  \begin{center}
    {\LARGE\bf a general proof technique}
  \end{center}

  \vspace*{1cm}

  \begin{center}
    {\bf Emmanuel Polonowski}
  \end{center}

  \vspace*{5cm}

  \begin{center}
    {\it April\ \ 2006}\\[1ex]
    {\large\bf TR--LACL--2006--5}
  \end{center}

  \vspace*{6cm}

  \hrule

  {\begin{center}
      \textbf{Laboratoire d'Algorithmique, Complexit\'e et Logique (LACL)}\\
	     {\small{\bf  D\'epartement d'Informatique}\\
	       {\bf Universit\'e Paris~12 -- Val~de~Marne, Facult\'e des Science et
		 Technologie}\\
	       61, Avenue du G\'en\'eral de Gaulle, 94010 Cr\'eteil cedex, France\\
	       Tel.: (33)(1) 45 17 16 47,\ \  Fax: (33)(1) 45 17 66 01}
  \end{center}}

\end{minipage}

  \clearpage

  \thispagestyle{empty}

  \vspace*{3cm}

  \begin{center}
    Laboratory of Algorithmics, Complexity and Logic (LACL)\\
    University Paris 12 (Paris Est)\\[1ex]
    Technical Report {\bf TR--LACL--2006--5}\\[2ex]
    E.~Polonowski.\\ {\it Deriving SN from PSN: a general proof technique}\\
  \end{center}

  \copyright\ \  E.~Polonowski,  April 2006.

  \clearpage

\setcounter{page}{1}

\maketitle

%  \vspace*{1cm}

%  \begin{center}
%    {\LARGE Deriving SN from PSN: a general proof technique}
%  \end{center}

%  \vspace*{1cm}

%  \begin{center}
%    {\bf Emmanuel Polonowski}
%  \end{center}

%  \vspace*{1cm}

\begin{abstract}
  In the framework of explicit substitutions there is two termination
  properties: preservation of strong normalization (PSN), and
  strong normalization (SN). Since there are not easily proved,
  only one of them is usually established (and sometimes none).
  We propose here a connection between them which helps to get SN
  when one already has PSN. For this purpose, we formalize
  a general proof technique of SN
  which consists in expanding substitutions into ``pure'' $\lambda$-terms
  and to inherit SN of the whole calculus by SN of the ``pure'' calculus
  and by PSN. We apply it successfully to a large set of calculi with
  explicit substitutions, allowing us to establish SN, or, at least, to
  trace back the failure of SN to that of PSN.
\end{abstract}

\tableofcontents

\section{Introduction}

Calculi with explicit substitutions were
introduced~\cite{AbaCarCurLev91} as a bridge
between $\lambda$-calculus~\cite{Chu41,Bar81} and concrete implementations of functional
programming languages. Those calculi intend to refine the evaluation
process by proposing reduction rules to deal with the substitution
mechanism --  a \textit{meta}-operation in the traditional $\lambda$-calculus.
It appears that, with those new rules, it was much harder (and sometimes impossible)
to get termination properties.

The two main termination properties of calculi with explicit substitutions are:
\begin{list}{\textbullet}{}
\item \textbf{Preservation of strong normalization} (PSN), which says that if a
  pure term (\textit{i.e.} without explicit substitutions)
  is strongly normalizing (\textit{i.e.} cannot be
  infinitely reduced) in the pure calculus (\textit{i.e.} the calculus
  without explicit substitutions),
  then this term is also strongly normalizing with respect to
  the calculus with explicit substitutions.
\item \textbf{Strong normalization} (SN), which says that, with respect to a
  typing system, every typed term is strongly normalizing in the
  calculus with explicit substitutions, \textit{i.e.} every terms in
  the subset of typed terms cannot be infinitely reduced.
\end{list}

These two properties are not redundant, and Fig.~\ref{fig:lambdaPSN}
shows the differences between them. PSN says that the horizontally
and diagonally hatched rectangle is included in the diagonally hatched rectangle.
SN says that the vertically hatched rectangle is included in
the diagonally hatched rectangle. Even if they work on a different set of terms,
there is a common part: the vertically and horizontally hatched rectangle, which
represent the typed pure terms.

\begin{figure}[htb]
\begin{center}
%\hspace*{-3cm}
  \includegraphics{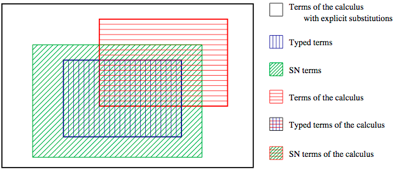}
\end{center}
\caption{Termination properties}
\label{fig:lambdaPSN}
\end{figure}

SN and PSN are both termination properties, although their proofs
are not always clearly related: sometimes
SN is shown independently of PSN (directly,
by simulation, \emph{etc.}, see for
example~\cite{CosKesPol03,DavGui03}),
sometimes SN proofs uses PSN (see for example~\cite{Blo97}).
We present here a general proof
technique of SN via PSN, initially suggested by H. Herbelin,
which uses that common part of typed pure terms.

\bigskip

More formally, we may introduce the following notations: we denote $\Lambda$ the set
of $\lambda$-terms, $\Lambda_T$ the set of typed $\Lambda$-terms with a given typing system $T$,
$\Lambda_{SN}$ the set of terminating $\Lambda$-terms ({\em i.e.} with a finite derivation tree);
we denote $\Lambda^X$, $\Lambda^X_T$, $\Lambda^X_{SN}$ the corresponding set for calculi with
eXplicit substitutions.

By definition, we have the following set inclusions:

$$
\Lambda_T \subset \Lambda \;\;\;\; \mbox{and} \;\;\;\; \Lambda_{SN} \subset \Lambda
$$
$$
\Lambda^X_T \subset \Lambda^X \;\;\;\; \mbox{and} \;\;\;\; \Lambda^X_{SN} \subset \Lambda^X
$$
$$
\Lambda \subset \Lambda^X \;\;\;\; \mbox{and} \;\;\;\; \Lambda_T \subset \Lambda^X_T
$$

The usual strong normalisation property of typed $\lambda$-calculus gives

$$
\Lambda_T \subset \Lambda_{SN}
$$

%We then get the diagram presented in Fig.~\ref{fig:diag}, where arrows denote set inclusions.

%\begin{figure}[htb]
%\begin{center}
%  \includegraphics{Figures/diagramme.pdf}
%\end{center}
%\caption{Termination properties (cont.)}
%\label{fig:diag}
%\end{figure}

As regard to calculi with explicit substitutions, we have the following properties.
At first, the property PSN gives

$$
\Lambda_{SN} \subset \Lambda^X_{SN}
$$

At last, the strong normalization property of typed $\Lambda^X$-terms completes
% the diagram (Fig.~\ref{fig:diag2})
with the following inclusion:

$$
\Lambda^X_T \subset \Lambda^X_{SN}
$$

%\begin{figure}[htb]
%\begin{center}
%  \input{Figures/diagramme2.eepic}
%\end{center}
%\caption{Termination properties (cont.)}
%\label{fig:diag2}
%\end{figure}

\bigskip

In the following section, we formalize a proof technique that exploits this diagram
and in
the remaining sections we apply this technique to a set of calculi.
This set has been chosen
for the variety of their definitions: with or without De Bruijn indices,
unary or multiple substitutions, with or without composition of
substitutions, and even a symmetric non-deterministic calculus.
In the last section, we briefly talk about perspectives in this
framework.

\section{Proof Technique}

The idea of this technique is the following. Let $t$ be a typed term
with explicit substitutions for which we want to show termination.
With the help of its typing judgment, we build a typed pure term
$t'$ which can be reduced to $t$. For that purpose, we expand
the substitutions of $t$ into redexes. We call this expansion $\mathit{Ateb}$
(the opposite of $\mathit{Beta}$ which is usually the name of the rule
which creates explicit substitutions).
Then, with SN of the pure calculus and PSN, we can export the strong normalization
of $t'$ (in the pure calculus) to $t$ (in the calculus with explicit substitutions).

\bigskip

In practice, this sketch will only apply in some cases, and some others
will require some adjustment to this technique. For our technique to work, we need
that the $\mathit{Ateb}$ expansion satisfies some properties. The first one
is always easily checked.

\begin{property}[Preservation of typability]\label{prop:type-direct}
  If $t$ is typable, with respect to a typing system $T$,
  in the calculus with explicit substitution,
  then $\mathit{Ateb}(t)$ is typable, with respect to a typing system $T'$ (possibly
  $T'=T$) in the pure calculus.  
\end{property}

Only some calculi can exhibit an $\mathit{Ateb}$ function which satisfies the
second one.

\begin{property}[Initialization]\label{prop:init-direct}
  $\mathit{Ateb}(t)$ reduces to $t$ in zero or more steps in the calculus with
  explicit substitutions.
\end{property}

If we can get it, then we use the direct proof to be presented in section~\ref{sect:direct}.
Otherwise, we need to use the simulation proof to be presented in section~\ref{sect:simul}.

\subsection{Direct proof}\label{sect:direct}

We can immediately establish the theorem.

\begin{theorem}\label{th:direct}
For all typing systems $T$ and $T'$ such that, in the pure calculus,
all typable terms with respect to $T$ are strongly
normalizing, if there exists a function $\mathit{Ateb}$ from explicit
substitution terms to pure terms satisfying
properties~\ref{prop:type-direct} and~\ref{prop:init-direct}
then PSN implies SN (with respect to $T'$).
\end{theorem}

\begin{proof}
  For every typed term $t$ of the calculus with explicit substitution, $\mathit{Ateb}(t)$
  is a pure typed term (by property~\ref{prop:type-direct}). By
  the strong normalization hypothesis of the typed pure
  calculus, we have $\mathit{Ateb}(t)\in\Lambda_{SN}$.
  By hypothesis of PSN
  we obtain that $\mathit{Ateb}(t)$ is in $\Lambda^X_{SN}$. By
  property~\ref{prop:init-direct}, we get $\mathit{Ateb}(t)\ra^* t$, which
  gives us directly $t\in\Lambda^X_{SN}$.
\end{proof}

\subsection{Simulation proof}\label{sect:simul}

We must relax some constraints on $\mathit{Ateb}$. We will try to find
an expansion of $t$ to $t'$ such that $t'$ reduces to a term $u$
and there exists a relation $\mathcal{R}$ with $u\mathcal{R} t$. The chosen relation
must, in addition, enable a simulation of the reductions of $t$
by the reduction of $u$. If it is possible, we can infer
strong normalization of $t$ from strong normalization of $u$.

To proceed with the simulation, we first split the reduction
rules of the calculus with explicit substitutions into two
disjoints sets. The set $R_1$ contains rules which are
trivially terminating, and $R_2$ contains the others. Secondly,
we build a relation $\mathcal{R}$ which satisfies the following properties.

\begin{property}[Initialisation]\label{prop:init-simul}
  For every typed term $t$, there exists a term $u\mathcal{R} t$ such that
  $\mathit{Ateb}(t)$ reduces in 0 or more steps to $u$ in the calculus
  with explicit substitutions.
\end{property}

\begin{property}[Simulation $^*$]\label{prop:simul-*}
  For every term $t$, if $t\ra_{R_1} t'$ then, for every $u\mathcal{R} t$,
  there exists $u'$ such that $u\ra^*u'$ and $u'\mathcal{R} t'$.
\end{property}

\begin{property}[Simulation $^+$]\label{prop:simul-+}
  For every term $t$, if $t\ra_{R_2} t'$ then, for every $u\mathcal{R} t$,
  there exists $u'$ such that $u\ra^+u'$ and $u'\mathcal{R} t'$.
\end{property}

We display those properties as diagrams :
$$
\begin{array}{r@{\ }l@{\ }c}
  \multicolumn{3}{c}{\mbox{Initialisation}} \\ \\
  & & t \\
  & \swarrow & \mathcal{R} \\
  \mathit{Ateb}(t) & \ra^* & u
\end{array}
\hspace{2cm}
\begin{array}{ccc}
  \multicolumn{3}{c}{\mbox{Simulation $^*$}} \\ \\
  t & \;\;\ra_{R_1}\;\; & t' \\
  \mathcal{R} & & \mathcal{R} \\
  u & \ra^* & u'
\end{array}
\hspace{2cm}
\begin{array}{ccc}
  \multicolumn{3}{c}{\mbox{Simulation $^+$}} \\ \\
  t & \;\;\ra_{R_2}\;\; & t' \\
  \mathcal{R} & & \mathcal{R} \\
  u & \ra^+ & u'
\end{array}
$$

With this material, we can establish the theorem.

\begin{theorem}\label{th:simul}
For all typing systems $T$ and $T'$ such that, in the pure calculus,
all typable terms with respect to $T$ are strongly
normalizing, if there exists a function $\mathit{Ateb}$ from explicit
substitution terms to pure terms and a relation $\mathcal{R}$ on
explicit substitutions terms satisfying
properties~\ref{prop:type-direct}, \ref{prop:init-simul},
\ref{prop:simul-*} and~\ref{prop:simul-+}
then PSN implies SN (with respect to $T'$).
\end{theorem}

\begin{proof}
  We prove it by contradiction. Let $t$ be a typed term with explicit
  substitutions which can be infinitely reduced. By
  property~\ref{prop:init-simul} there exists a term $u$ such that
  $\mathit{Ateb}(t)\ra^* u$, and $\mathit{Ateb}(t)$ is a pure typed term
  (by property~\ref{prop:type-direct}). By the strong
  normalization hypothesis of the typed pure
  calculus, we have $\mathit{Ateb}(t)\in\Lambda_{SN}$.
  By hypothesis of PSN
  we obtain that $\mathit{Ateb}(t)$ is in $\Lambda^X_{SN}$ and
  it follows that $u\in\Lambda^X_{SN}$.

  By property~\ref{prop:init-simul}, we also have $u\mathcal{R} t$, and,
  with properties~\ref{prop:simul-*} and~\ref{prop:simul-+},
  we can build an infinite reduction from $u$, contradicting the
  strong normalization of $u$. 
\end{proof}

\section{$\lambda{\tt x}$-calculus}

The $\lambda{\tt x}$-calculus~\cite{BloRos95,BloGeu99} is probably the simplest
calculus with explicit substitutions.
It only makes the substitution explicit.
Since this calculus provides no rules to deal with substitutions
composition, it preserves strong normalization. It is for this calculus
that the technique has been originate used by Herbelin. Therefore, we can
without surprises apply the direct proof to get strong
normalization.

\subsection{Definition}

Terms of the $\lambda{\tt x}$-calculus are given by the following grammar:

$$
t ::= x\ |\ (t\ t)\ |\ \lambda x.t\ |\ t[t/x]
$$

Here follows the reduction rules:

$$
\begin{array}{rcl}
  (\lambda x.t) u & \ra_{\mathit{Beta}} & t[u/x] \\
  (t\ u)[v/x] & \ra_{\mathit{App}} & (t[v/x])\ (u[v/x]) \\
  (\lambda x.t)[u/y] & \ra_{\mathit{Lambda}} & \lambda x.(t[u/y]) \\
  x[t/x] & \ra_{\mathit{Var1}} & t \\
  y[t/x] & \ra_{\mathit{Var2}} & y \\
\end{array}
$$

The rule $\mathit{Lambda}$ is applied modulo $\alpha$-conversion of the bound variable $x$.

Here follows the typing rules:

\begin{center}
  \begin{tabular}{cc}
    \infer{\Gamma,x:A\vdash x:A}{} &
    \infer
      {\Gamma\vdash t[u/x]:A}
      {\Gamma\vdash u:B & \Gamma,x:B\vdash t:A} \\ & \\
    \infer
      {\Gamma\vdash (t\ u):A}
      {\Gamma\vdash t:B\ra A & \Gamma\vdash u:B} &
    \infer
      {\Gamma\vdash \lambda x.t:B\ra A}
      {\Gamma,x:B\vdash t:A} \\
  \end{tabular}
\end{center}

\subsection{Strong normalisation proof}

We define the $\mathit{Ateb}$ function as follows:

\begin{center}
  \begin{tabular}{lll}
    $\mathit{Ateb}(x)$ & $=$ & $x$ \\
    $\mathit{Ateb}(t\ u)$ & $=$ & $\mathit{Ateb}(t)\ \mathit{Ateb}(u)$ \\
    $\mathit{Ateb}(\lambda x.t)$ & $=$ & $\lambda x.\mathit{Ateb}(t)$ \\
    $\mathit{Ateb}(t[u/x])$ & $=$ & $(\lambda x.\mathit{Ateb}(t))\ \mathit{Ateb}(u)$ \\
  \end{tabular}
\end{center}

Remark that $\mathit{Ateb}$ performs the exact reverse rewriting of the rule $\mathit{Beta}$.
It straightforwardly follows that if $t'=\mathit{Ateb}(t)$ then
$t'\ra_{\mathit{Beta}}^*t$ and $\mathit{Ateb}(t)$ does not contain any substitutions.

We check that the $\mathit{Ateb}(t)$ is typable.

\begin{lemma}
  $$\Gamma\vdash t:A \;\;\Ra\;\; \Gamma\vdash \mathit{Ateb}(t):A$$
\end{lemma}

\begin{proof}
  By induction on the typing derivation of $t$. The only non-trivial case is that of
  substitutions. We have $t=u[v/x]$ and

  $$
  \infer
    {\Gamma\vdash u[v/x]:A}
    {\Gamma\vdash v:B & \Gamma,x:B\vdash u:A}
  $$

  By induction hypothesis, we have $\Gamma,x:B\vdash \mathit{Ateb}(u):A$ and
  $\Gamma\vdash \mathit{Ateb}(v):B$. We can build the typing derivation
  of $\mathit{Ateb}(t)=\lambda x.\mathit{Ateb}(u))\ \mathit{Ateb}(v)$
  as follows

  $$
  \infer
    {\Gamma\vdash (\lambda x.\mathit{Ateb}(u))\ \mathit{Ateb}(v):A}
    {\infer
      {\Gamma\vdash \lambda x.\mathit{Ateb}(u):B\ra A}
      {\Gamma,x:B\vdash \mathit{Ateb}(u):A}
      & \Gamma\vdash \mathit{Ateb}(v):B}
  $$
\end{proof}

We can apply Theorem~\ref{th:direct}.

\begin{corollary}
  Since the $\lambda{\tt x}$-calculus enjoys PSN~\cite{BloRos95} and the
  $\lambda$-calculus enjoys SN of simply-typed terms~\cite{Kri90}, we conclude
  that the $\lambda{\tt x}$-calculus enjoys SN of simply-typed terms.
\end{corollary}

\section{$\lambda\upsilon$-calculus}

The $\lambda\upsilon$-calculus~\cite{Les94,Les96} is the De Bruijn counterpart of $\lambda{\tt x}$.
As $\lambda{\tt x}$, it has no composition rules, and therefore satisfies PSN.
For this calculus, we must use the simulation proof to deal with
indices modification operators. We succeed to use it and
it is, as far as we know, the first proof of SN for a simply typed version
of $\lambda\upsilon$ (see~\cite{PolPhD04}).

\subsection{Definition}

Terms of $\lambda\upsilon$-calculus are given by the following grammar:

$$
\begin{array}{l}
  t ::= n\ |\ (t\ t)\ |\ \lambda t\ |\ t[s] \\
  s ::= a/\ |\ \Ua(s)\ |\ \ua \\
\end{array}
$$

Remark that a substitution is always build from a (possibly empty) list of
$\Ua$ followed by either a $t/$, or a $\ua$.
We will then write substitutions in a more general form: either $t[\Ua^i(t/)]$,
or $t[\Ua^i(\ua)]$, where $\Ua^i(s)$
denotes
$\underset{i}{\underbrace{\Ua(\Ua(...(\Ua}}(s))...))$.

Here follows the reduction rules:

$$
\begin{array}{lll}
  (\lambda t) u & \ra_{B} & t[u/] \\
  (t\ u)[s] & \ra_{\mathit{App}} & (t[s])\ (u[s]) \\
  (\lambda t)[s] & \ra_{\mathit{Lambda}} & \lambda (t[\Ua(s)]) \\
  1[t/] & \ra_{\mathit{FVar}} & t \\
  n+1[t/] & \ra_{\mathit{RVar}} & n \\
  1[\Ua(s)] & \ra_{\mathit{FVarLift}} & 1 \\
  n+1[\Ua(s)] & \ra_{\mathit{RVarLift}} & n[s][\ua] \\
  n[\ua] & \ra_{\mathit{VarShift}} & n+1 \\
\end{array}
$$

Here follows the typing rules (where $n=|\Gamma|+1$) :

\begin{center}
  \begin{tabular}{cc}
    \infer{\Gamma,A,\Delta\vdash n:A}{} &
    \infer
      {\Gamma\vdash t[s]:A}
      {\Gamma\vdash s\trg \Gamma' & \Gamma'\vdash t:A} \\ & \\
    \infer
      {\Gamma\vdash (t\ u):A}
      {\Gamma\vdash t:B\ra A & \Gamma\vdash u:B} &
    \infer
      {\Gamma\vdash \lambda t:B\ra A}
      {B,\Gamma\vdash t:A} \\
  \end{tabular}
\end{center}

\begin{center}
  \begin{tabular}{c@{\espace}c@{\espace}c}
    \infer
      {\Gamma\vdash t/\trg A,\Gamma}
      {\Gamma\vdash t:A} &
    \infer
      {A,\Gamma\vdash \ua\trg\ \Gamma}{} &
    \infer
      {A,\Gamma\vdash \Ua(s)\trg A,B,\Gamma}
      {\Gamma\vdash s\trg B,\Gamma} \\
  \end{tabular}
\end{center}

\subsection{Strong normalisation proof}\label{sect:lu-sn}

We define the $\mathit{Ateb}$ function as follows:

$$
\begin{array}{lll}
  \mathit{Ateb}(n) & = & n \\
  \mathit{Ateb}(t\ u) & = & \mathit{Ateb}(t)\ \mathit{Ateb}(u) \\
  \mathit{Ateb}(\lambda t) & = & \lambda \mathit{Ateb}(t) \\
  \mathit{Ateb}(t[u/]) & = & (\lambda \mathit{Ateb}(t))\ \mathit{Ateb}(u) \\
  \mathit{Ateb}(t[\Ua^i(u/)]) & = & (\lambda \Fliftcons{i}{\mathit{Ateb}(t)})\
  \Fshift{i}{\mathit{Ateb}(u)} \\
  \mathit{Ateb}(t[\Ua^i(\ua)]) & = & \Fliftshift{i}{\mathit{Ateb}(t)} \\
\end{array}
$$

\begin{example}
  For instance, if we suppose that for any $tt$ among $t$, $u$, $v$,
  $w$ we have $tt=\mathit{Ateb}(tt)$, then we get
  $$
  \mathit{Ateb}((t[u/]\ v[\Ua(\Ua(\Ua(w/)))])[\Ua(\Ua(\ua))]) \;\;\; = \;\;\;
  \Fliftshift{2}{((\lambda t)u)\ ((\lambda\Fliftcons{3}{v})\Fshift{3}{w})}
  $$
\end{example}

Where $\Fliftcons{i}{t})$, $\Fshift{i}{t}$ and $\Fliftshift{i}{t}$ are functions
that we will define in the sequel. The intuition about those function is the following:
substitutions perform some re-indexing of terms upon which they are applied, those functions
intend to anticipate this re-indexing. To understand the necessity of those functions,
let us look at some typing derivation. To begin with, we take $t[\Ua^i(u/)]$, where
$\Delta=D_i,...,D_1$ ($i=|\Delta|$) :

$$
\infer
  {\Delta,\Gamma\vdash t[\Ua^i(u/)]:A}
  {\infer
    {D_i,D_{i-1},...,D_1,\Gamma\vdash\Ua^i(u/)\trg D_i,D_{i-1},...,D_1,B,\Gamma}
    {\infer
      {D_{i-1},...,D_1,\Gamma\vdash\Ua^{i-1}(u/)\trg
        D_{i-1},...,D_1,B,\Gamma}
      {\infer
        {\vdots}
        {\infer
          {D_1,\Gamma\vdash\Ua(u/)\trg D_1,B,\Gamma}
          {\infer
            {\Gamma\vdash u/\trg B,\Gamma}
            {\Gamma\vdash u:B}}}}} &
   \Delta,B,\Gamma\vdash t:A}
$$

We would like to type a term of the form $(\lambda t') u'$,
that is

$$
\infer
  {\Delta,\Gamma\vdash(\lambda t') u':A}
  {\infer
    {\Delta,\Gamma\vdash\lambda t':B\ra A}
    {B,\Delta,\Gamma\vdash t':A} &
   \Delta,\Gamma\vdash u':B}
$$

The problem is to build a term $t'$ from $t$ which would be
typeable in the environment $B,\Delta,\Gamma$ instead of
$\Delta,B,\Gamma$ and a term $u'$ from $u$ which would be
typeable in the environment $\Delta,\Gamma$ instead of $\Gamma$. This
is exactly the work of the functions $\Fliftcons{i}{\ }$ and
$\Fshift{i}{\ }$ respectively. Look now at the typing derivation
of $t[\Ua^i(\ua)]$, where $\Delta=D_i,...,D_1$ ($i=|\Delta|$) :

$$
\infer
  {\Delta,B,\Gamma\vdash t[\Ua^i(\ua)]:A}
  {\infer
    {D_i,D_{i-1},...,D_1,B,\Gamma\vdash\Ua^i(\ua)\trg D_i,D_{i-1},...,D_1,\Gamma}
    {\infer
      {D_{i-1},...,D_1,B,\Gamma\vdash\Ua^{i-1}(\ua)\trg
        D_{i-1},...,D_1,\Gamma}
      {\infer
        {\vdots}
        {\infer
          {D_1,B,\Gamma\vdash\Ua(\ua)\trg D_1,\Gamma}
          {B,\Gamma\vdash\ua\trg\Gamma}}}} &
   \Delta,\Gamma\vdash t:A}
$$

The problem here is to build a term $t'$ from
$t$ which would be
typeable in the environment $\Delta,B,\Gamma$ instead of
$\Delta,\Gamma$. This is done by the function $\Fliftshift{i}{\ }$. We can state
the property that should verify those functions.

\begin{property}\label{p:fct-lu}
  For any term $t$ we have (with $i=|\Delta|$) :
  \begin{itemize}
  \item $\Gamma\vdash t:A \;\;\;\Ra\;\;\;
    \Delta,\Gamma\vdash\Fshift{i}{t}:A$
  \item $\Delta,B,\Gamma\vdash t:A \;\;\;\Ra\;\;\;
    B,\Delta,\Gamma\vdash\Fliftcons{i}{t}:A$
  \item $\Delta,\Gamma\vdash t:A \;\;\;\Ra\;\;\;
    \Delta,B,\Gamma\vdash\Fliftshift{i}{t}:A$
  \end{itemize}
\end{property}

We can then check that the term obtained by the function $\mathit{Ateb}$ is typeable.

\begin{lemma}
  $$\Gamma\vdash t:A \;\;\Ra\;\; \Gamma\vdash\mathit{Ateb}(t):A$$
\end{lemma}

\begin{proof}
  By induction on the typing derivation of $t$.

  \begin{itemize}
  \item $t=n$ and

    $$
    \infer{\Gamma,A,\Delta\vdash n:A}{}
    $$

    We then have $\mathit{Ateb}(t)=n$ and the same typing derivation.

  \item $t=(u\ v)$ and

    $$
    \infer
      {\Gamma\vdash (u\ v):A}
      {\Gamma\vdash u:B\ra A & \Gamma\vdash v:B}
    $$

    By induction hypothesis, we have $\Gamma\vdash \mathit{Ateb}(u):B\ra A$ and
    $\Gamma\vdash \mathit{Ateb}(v):B$. We can type $\mathit{Ateb}(t)=\mathit{Ateb}(u)\
    \mathit{Ateb}(v)$ as follows

    $$
    \infer
      {\Gamma\vdash (\mathit{Ateb}(u)\ \mathit{Ateb}(v)):A}
      {\Gamma\vdash \mathit{Ateb}(u):B\ra A & \Gamma\vdash \mathit{Ateb}(v):B}
    $$

  \item $t=\lambda u$ and

    $$
    \infer
      {\Gamma\vdash \lambda u:B\ra A}
      {B,\Gamma\vdash u:A}
    $$

    By induction hypothesis, we have $\Gamma,x:B\vdash \mathit{Ateb}(u):A$. We can type
    $\mathit{Ateb}(t)=\lambda \mathit{Ateb}(u)$ as follows

    $$
    \infer
      {\Gamma\vdash \lambda \mathit{Ateb}(u):A}
      {B,\Gamma\vdash \mathit{Ateb}(u):A}
    $$

  \item Cases for $t=u[\Ua^i(v/)]$ and $t=u[\Ua^i(\ua)]$ are treated as discussed above,
    using property~\ref{p:fct-lu}.
  \end{itemize}
\end{proof}

Of course, for any $t$, $\mathit{Ateb}(t)$ does not contain any substitutions.

\subsubsection{Functions definitions}

The function $\Fliftshift{i}{t}$ performs the re-indexing of $t$
as if a substitution $[\Ua^i(\ua)]$ has been propagated. Since it is applied
to terms obtained by the $\mathit{Ateb}$ function, only terms without substitutions
are concerned.

Here follows its definition:

$$
\begin{array}{llll}
  \Fliftshift{i}{n} & = & n+1 & \mbox{if } n>i \\
  \Fliftshift{i}{n} & = & n & \mbox{if } n\leq i \\
  \Fliftshift{i}{t\ u} & = & \Fliftshift{i}{t}\
  \Fliftshift{i}{u} & \\
  \Fliftshift{i}{\lambda t} & = & \lambda\Fliftshift{i+1}{t} & \\
\end{array}
$$

The function $\Fshift{i}{t}$ performs the re-indexing of $t$ as if $i$
substitutions $[\ua]$ have been propagated. We can define it from with the
help of the function $\Fliftshift{i}{t}$ :

$$
\Fshift{i}{t} =
\underset{i}{\underbrace{\Fliftshift{0}{\Fliftshift{0}{...\Fliftshift{0}{t}}}}}
$$

When this function is applied to a variable, we obtain $\Fshift{i}{n}=n+i$.

The function $\Fliftcons{i}{t}$ prepares a term $t$ to be applied to a
substitution that has lost its $\Ua$. It deals also with substitution-free
terms.

Here follows its definition:

$$
\begin{array}{llll}
  \Fliftcons{i}{n} & = & n & \mbox{si } n>i+1 \\
  \Fliftcons{i}{n} & = & 1 & \mbox{si } n=i+1 \\
  \Fliftcons{i}{n} & = & n+1 & \mbox{si } n\leq i \\
  \Fliftcons{i}{t\ u} & = & \Fliftcons{i}{t}\
  \Fliftcons{i}{u} & \\
  \Fliftcons{i}{\lambda t} & = & \lambda\Fliftcons{i+1}{t} & \\
\end{array}
$$

Indices are transformed as follows: since we have deleted $i$ $\Ua$,
the index $i+1$ must become $1$. To reflect this change, every index $j$
smaller than $i+1$ must become $j+1$. The others are let unchanged.

\bigskip

Here follows several useful properties.

\begin{property}\label{p:funshift+1-lu}
  For all $t$, $u$, $i$, $j$, We have

  $$
  \Fshift{i}{t} = \Fshift{j}{u} \;\;\;\;\; \Ra \;\;\;\;\;
  \Fshift{i+1}{t} = \Fshift{j+1}{u}
  $$
\end{property}

\begin{proof}
  Indeed,

  $$
  \Fshift{i}{t} = \Fshift{j}{u} \;\;\;\;\; \Ra \;\;\;\;\;
  \Fliftshift{0}{\Fshift{i}{t}} = \Fliftshift{0}{\Fshift{j}{u}}
  $$
\end{proof}

\begin{property}\label{p:funshift-comp-lu}
  For all $n$ and $i$, we have

  $$
  \Fliftshift{i+1}{n} \;\;\;\;\; = \;\;\;\;\; \Fshift{1}{\Fliftshift{i}{n-1}}
  $$
\end{property}

\begin{proof}
  We calculate the values accordingly to $n$ and $i$.

  \begin{itemize}
  \item if $n>i+1$ then $\Fliftshift{i+1}{n}=n$,
    $\Fliftshift{i}{n-1}=n-1$ and $\Fshift{1}{n-1}=n$.
  \item if $n\leq i+1$ then $\Fliftshift{i+1}{n}=n+1$,
    $\Fliftshift{i}{n-1}=n$ and $\Fshift{1}{n}=n+1$.
  \end{itemize}
\end{proof}

\begin{property}\label{p:funcons-comp-lu}
  For all $n>1$ and $i$, we have

  $$
  \Fliftcons{i+1}{n} \;\;\;\;\; = \;\;\;\;\; \Fliftshift{1}{\Fliftcons{i}{n-1}}
  $$
\end{property}

\begin{proof}
  We calculate the values accordingly to $n$ and $i$.

  \begin{itemize}
  \item if $n>i+2$ then $\Fliftcons{i+1}{n}=n$,
    $\Fliftcons{i}{n-1}=n-1$ and $\Fliftshift{1}{n-1}=n$.
  \item if $n=i+2$ then $\Fliftcons{i+1}{n}=1$,
    $\Fliftcons{i}{n-1}=1$ and $\Fliftshift{1}{1}=1$.
  \item if $n<i+2$ then $\Fliftcons{i+1}{n}=n+1$,
    $\Fliftcons{i}{n-1}=n$ and $\Fliftshift{1}{n}=n+1$.
  \end{itemize}
\end{proof}

\begin{example}
  We apply those function to our example, and we obtain
  $$
  \Fliftshift{2}{((\lambda t)u)\
    ((\lambda\Fliftcons{3}{v})\Fshift{3}{w})} \;\;\; = \;\;\;
  ((\lambda \Fliftshift{3}{t})\Fliftshift{2}{u}) {\ }
  ((\lambda \Fliftshift{3}{\Fliftcons{3}{v}})\Fliftshift{2}{\Fshift{3}{w}})
  $$
\end{example}

We can now prove Property~\ref{p:fct-lu}.

\begin{proof}
  \begin{itemize}
  \item $\Delta,\Gamma\vdash t:A \;\;\;\Ra\;\;\;
    \Delta,B,\Gamma\vdash\Fliftshift{i}{t}:A$.
    By induction on $t$.
    \begin{itemize}
    \item $t=n$ with $n\leq i$: $\Fliftshift{i}{t} = n$. We have

      $$
      \infer{\Delta_1,A,\Delta_2,\Gamma\vdash n:A}{}
      $$

      with $n=|\Delta_1|+1$. We conclude with the following typing derivation

      $$
      \infer{\Delta_1,A,\Delta_2,B,\Gamma\vdash n:A}{}
      $$

    \item $t=n$ with $n>i$: $\Fliftshift{i}{t} = n+1$. We have

      $$
      \infer{\Delta,\Gamma_1,A,\Gamma_2\vdash n:A}{}
      $$

      With $n=|\Delta|+|\Gamma_1|+1$. We the get
      $n+1=|\Delta|+|\Gamma_1|+1+1$ and

      $$
      \infer{\Delta,B,\Gamma_1,A,\Gamma_2\vdash n:A}{}
      $$

    \item $t=(u\ v)$: $\Fliftshift{i}{t} = (\Fliftshift{i}{u}\
      \Fliftshift{i}{v})$ and we conclude by applying twice the induction hypothesis.
    \item $t=\lambda u$ (with $A=C\ra D$) : $\Fliftshift{i}{t} =
      \lambda\Fliftshift{i+1}{u}$. We have

      $$
      \infer
        {\Delta,\Gamma\vdash\lambda u:C\ra D}
        {C,\Delta,\Gamma\vdash u:D}
      $$

      By induction hypothesis, we have
      $C,\Delta,B,\Gamma\vdash\Fliftshift{i+1}{u}:D$ and we can build the
      following typing derivation

      $$
      \infer
        {\Delta,B,\Gamma\vdash\lambda \Fliftshift{i+1}{u}:C\ra D}
        {C,\Delta,B,\Gamma\vdash \Fliftshift{i+1}{u}:D}
      $$
    \end{itemize}

  \item $\Gamma\vdash t:A \;\;\;\Ra\;\;\;
    \Delta,\Gamma\vdash\Fshift{i}{t}:A$. By induction hypothesis on $t$.
    \begin{itemize}
    \item $t=n$:
      $\Fshift{i}{n} =
      \underset{i}{\underbrace{\Fliftshift{0}{\Fliftshift{0}{...\Fliftshift{0}{n}}}}}
      = n+i$. We have

      $$
      \infer{\Gamma_1,A,\Gamma_2\vdash n:A}{}
      $$

      with $n=|\Gamma_1|+1$. Since $i=|\Delta|$, we get
      $n+i=|\Gamma_1|+|\Delta|+1$ and

      $$
      \infer{\Delta,\Gamma_1,A,\Gamma_2\vdash n+i:A}{}
      $$

    \item $t=(u\ v)$: $\Fshift{i}{t} = (\Fshift{i}{u}\
      \Fshift{i}{v})$ and we conclude by applying twice the induction hypothesis.
    \item $t=\lambda u$ (with $A=C\ra D$): $\Fshift{i}{t} =
      \underset{i}{\underbrace{\Fliftshift{0}{\Fliftshift{0}{...\Fliftshift{0}{\lambda u}}}}}
      =
      \lambda\underset{i}{\underbrace{\Fliftshift{1}{\Fliftshift{1}{...\Fliftshift{1}{u}}}}}$.

      We get

      $$
      \infer
        {\Gamma\vdash\lambda u:C\ra D}
        {C,\Gamma\vdash u:D}
      $$

      By the item above, we have

      $$
      \begin{array}{c}
        C,\Gamma\vdash u:D \\
        \Da \\
        C,E_1,\Gamma\vdash \Fliftshift{1}{u}:D \\
        \Da \\
        C,E_2,E_1,\Gamma\vdash \Fliftshift{1}{\Fliftshift{1}{u}}:D \\
        \Da \\ \vdots \\ \Da \\
        C,E_i,...,E_1,\Gamma\vdash
        \underset{i}{\underbrace{\Fliftshift{1}{\Fliftshift{1}{...\Fliftshift{1}{u}}}}} \\
      \end{array}
      $$

      with $\Delta=E_i,...,E_1$. We can then build the following typing derivation

      $$
      \infer
        {\Delta,\Gamma\vdash\Fshift{i}{\lambda u}:C\ra D}
        {C,\Delta,\Gamma\vdash
          \underset{i}{\underbrace{\Fliftshift{1}{\Fliftshift{1}{...\Fliftshift{1}{u}}}}}:D}
      $$
    \end{itemize}

  \item $\Delta,B,\Gamma\vdash t:A \;\;\;\Ra\;\;\;
    B,\Delta,\Gamma\vdash\Fliftcons{i}{t}:A$. By induction on $t$.
    \begin{itemize}
    \item $t=n$ with $n>i+1$: $\Fliftcons{i}{t} = n$. We have

      $$
      \infer{\Delta,B,\Gamma_1,A,\Gamma_2\vdash n:A}{}
      $$

      with $n=|\Delta|+1+|\Gamma_1|+1$. We conclude with the following typing derivation

      $$
      \infer{B,\Delta,\Gamma_1,A,\Gamma_2\vdash n:A}{}
      $$

    \item $t=n$ with $n=i+1$: $\Fliftcons{i}{t} = 1$. We have

      $$
      \infer{\Delta,B,\Gamma\vdash n:B}{}
      $$

      with $n=i+1=|\Delta|+1$. We conclude with the following typing derivation

      $$
      \infer{B,\Delta,\Gamma\vdash 1:A}{}
      $$

    \item $t=n$ with $n\leq i$: $\Fliftcons{i}{t} = n+1$. We have

      $$
      \infer{\Delta_1,A,\Delta_2,B,\Gamma\vdash n:A}{}
      $$

      with $n=|\Delta_1|+1$. We then get
      $n+1=|\Delta_1|+1+1$ and

      $$
      \infer{B,\Delta_1,A,\Delta_2,\Gamma\vdash n:A}{}
      $$

    \item $t=(u\ v)$ : $\Fliftcons{i}{t} = (\Fliftcons{i}{u}\
      \Fliftcons{i}{v})$ and we conclude by applying twice the induction hypothesis.
    \item $t=\lambda u$ (with $A=C\ra D$) : $\Fliftcons{i}{t} =
      \lambda\Fliftcons{i+1}{u}$. We have

      $$
      \infer
        {\Delta,B,\Gamma\vdash\lambda u:C\ra D}
        {C,\Delta,B,\Gamma\vdash u:D}
      $$

      By induction hypothesis, we have
      $C,\Delta,\Gamma\vdash\Fliftcons{i+1}{u}:D$ and we can build the following
      typing derivation

      $$
      \infer
        {\Delta,\Gamma\vdash\lambda \Fliftcons{i+1}{u}:C\ra D}
        {C,\Delta,\Gamma\vdash \Fliftcons{i+1}{u}:D}
      $$
    \end{itemize}
  \end{itemize}
\end{proof}

\subsubsection{Definition of the relation $\rel$}

The function $\mathit{Ateb}$ erases the substitutions $[\Ua^i(\ua)]$ and we will
not be able to recover them by reducing the term obtained, as is shown for the
following example.

\begin{example}
  We continue with our example, we get
  $$
  \begin{array}{c}
  ((\lambda \Fliftshift{3}{t})\Fliftshift{2}{u}) {\ }
  ((\lambda
  \Fliftshift{3}{\Fliftcons{3}{v}})\Fliftshift{2}{\Fshift{3}{w}}) \\
  \ra_B\ra_B \\
  \Fliftshift{3}{t}[\Fliftshift{2}{u}/] {\ }
  \Fliftshift{3}{\Fliftcons{3}{v}}[\Fliftshift{2}{\Fshift{3}{w}}/]
  \end{array}
  $$
\end{example}

We must use the proof by simulation. To perform this simulation, we need a new
function $\ol{t}$ which performs the re-indexing of the erased substitutions.

$$
\begin{array}{lll}
  \ol{n} & = & n \\
  \ol{t\ u} & = & \ol{t}\ \ol{u} \\
  \ol{\lambda t} & = & \lambda \ol{t} \\
  \ol{t[u/]} & = & \ol{t}[\ol{u}/] \\
  \ol{t[\Ua^i(u/)]} & = & \Fliftcons{i}{\ol{t}}[\Fshift{i}{\ol{u}}/] \\
  \ol{t[\Ua^i(\ua)]} & = & \Fliftshift{i}{\ol{t}} \\
\end{array}
$$

This function will deal with terms that might contain substitutions. We need then
to extend their definition. By the way, since the function $\ol{t}$ removes from
$t$ the $\Ua$ and $\ua$, we will restrain our-self to the simple substitution case:

$$
\begin{array}{llll}
  \Fliftshift{i}{t[u/]} &=& \Fliftshift{i+1}{t}[\Fliftshift{i}{u}/] & \\
  \Fliftcons{i}{t[u/]} &=& \Fliftcons{i+1}{t}[\Fliftcons{i}{u}/] & \\
\end{array}
$$

The function $\ol{\cdot}$ commute with the other function, as stated in the following lemmas.

\begin{lemma}\label{l:commute-ol-fls}
  for all $i$ and $t$ (without $\Ua$ and $\ua$) we have

  $$
  \ol{\Fliftshift{i}{t}} \;\;\;\; = \;\;\;\; \Fliftshift{i}{\ol{t}}
  $$
\end{lemma}

\begin{proof}
  By induction on $t$.

  \begin{itemize}
  \item If $t=n$, then $\Fliftshift{i}{t}=n'$, $\ol{n'}=n'$ on one side, and
    $\ol{n} = n$ on the other side.
  \item In all the other cases, we conclude by induction hypothesis.
  \end{itemize}
\end{proof}

\begin{lemma}\label{l:commute-ol-flc}
  For all $i$ and $t$ (without $\Ua$ and $\ua$) we have

  $$
  \ol{\Fliftcons{i}{t}} \;\;\;\; = \;\;\;\; \Fliftcons{i}{\ol{t}}
  $$
\end{lemma}

\begin{proof}
  By induction on $t$.

  \begin{itemize}
  \item If $t=n$, then $\Fliftcons{i}{t}=n'$, $\ol{n'}=n'$ on one side, and
    $\ol{n} = n$ on the other side.
  \item In all the other cases, we conclude by induction hypothesis.
  \end{itemize}
\end{proof}

\begin{lemma}\label{l:commute-ol-fs}
  For all $i$ and $t$ (without $\Ua$ and $\ua$) we have

  $$
  \ol{\Fshift{i}{t}} \;\;\;\; = \;\;\;\; \Fshift{i}{\ol{t}}
  $$
\end{lemma}

\begin{proof}
  This is a direct consequence of Lemma~\ref{l:commute-ol-fls}.
\end{proof}

We can check that this function is correct w.r.t. our example.

\begin{example}
  Here is the final term obtain for our example:
  $$
  \begin{array}{c}
    \ol{\Fliftshift{3}{t}[\Fliftshift{2}{u}/] {\ }
      \Fliftshift{3}{\Fliftcons{3}{v}}[\Fliftshift{2}{\Fshift{3}{w}}/]} \\
    = \\
    \Fliftshift{3}{\ol{t}}[\Fliftshift{2}{\ol{u}}/] {\ }
      \Fliftshift{3}{\Fliftcons{3}{\ol{v}}}[\Fliftshift{2}{\Fshift{3}{\ol{w}}}/]
  \end{array}
  $$
  Here is the original term:
  $$
  \begin{array}{c}
    \ol{(t[u/]\ v[\Ua(\Ua(\Ua(w/)))])[\Ua(\Ua(\ua))]} \\
    = \\
    \Fliftshift{2}{\ol{t}[\ol{u}/]\ \Fliftcons{3}{\ol{v}}[\Fshift{3}{\ol{w}}/]} \\
    = \\
    \Fliftshift{3}{\ol{t}}[\Fliftshift{2}{\ol{u}}/] {\ }
      \Fliftshift{3}{\Fliftcons{3}{\ol{v}}}[\Fliftshift{2}{\Fshift{3}{\ol{w}}}/]
  \end{array}
  $$
\end{example}

We also need an order relation on the skeleton of terms. We want that $t\skelin t'$
if and only if $t$ contains $[\ua]$ and $\Ua$ only where $t'$ contains them also.
We formalize this definition as follows:

$$
\begin{array}{ccc}
  \mbox{for all } n\mbox{ and } m & & n\skelin m \\
  t\skelin t' \mbox{ and } u\skelin u' & \Ra & (t\ u)\skelin(t'\ u') \\
  t\skelin t' & \Ra & \lambda t\skelin\lambda t' \\
  t\skelin t' & \Ra & t\skelin t'[\ua] \\
  t\skelin t' \mbox{ and } s\skelin s' & \Ra & t[s]\skelin t'[s'] \\
\end{array}
$$

$$
\begin{array}{ccc}
  & & \ua\skelin \ua \\
  t\skelin t' & \Ra & t/\skelin t'/ \\
  s\skelin s' & \Ra & \Ua(s)\skelin\Ua(s') \\
  s\skelin s' & \Ra & s\skelin\Ua(s') \\
\end{array}
$$

\begin{example}
We have $t[\Ua(t'/)]\skelin t[\ua][\Ua(\Ua(\Ua(t'/)))]$.
\end{example}

From this relation and the function $\ol{t}$, we can build a relation to perform
our simulation. We note this relation $\rel$ and we define it as follows:

$$
t\rel t' \iff \ol{t} = \ol{t'} \mbox{ and } t\skelin t'
$$

Remark that we always have $t\rel t$. We can now initialize our simulation.

\begin{lemma}[Initialization]\label{l:Ateb-sim-lu}
  For all $t$, there exists $u$ such that $\mathit{Ateb}(t) \ra_B^* u$ and $u\rel t$.
\end{lemma}

\begin{proof}
  By induction on $t$.

  \begin{itemize}
  \item If $t=n$, then $\mathit{Ateb}(t) = n$ and it is enough to take $u=n$.
  \item If $t=t_1\ t_2$, then $\mathit{Ateb}(t) = \mathit{Ateb}(t_1)\ \mathit{Ateb}(t_2)$.
    By induction hypothesis, there exists $u_1$ and $u_2$ such that $\mathit{Ateb}(t_1)
    \ra_B^* u_1$ and $\mathit{Ateb}(t_2) \ra_B^* u_2$ with $u_1\rel t_1$ and $u_2\rel
    t_2$. We take $u=u_1\ u_2$.
  \item If $t=\lambda t'$, then we proceed as above using the induction hypothesis for $t'$.
  \item If $t=t'[\Ua^i(\ua)]$, then $\mathit{Ateb}(t) =
    \Fliftshift{i}{\mathit{Ateb}(t')}$. By induction hypothesis, there exists
    $u'$ such that $\mathit{Ateb}(t')\ra_B^* u'$ and $u'\rel t'$. We take
    $u=\Fliftshift{i}{u'}$ and we check that $u\rel t$, that is
    $\ol{u}=\ol{t}$ and $u\skelin t$. This last condition is trivial since
    $u'\skelin t'$. We calculate $\ol{u} = \ol{\Fliftshift{i}{u'}}$, which is equal to
    $\Fliftshift{i}{\ol{u'}}$ by Lemma~\ref{l:commute-ol-fls}. $\ol{t} =
    \ol{t'[\Ua^i(\ua)]} = \Fliftshift{i}{\ol{t'}}$, and we conclude since $\ol{u'}=\ol{t'}$.
  \item If $t=t_1[\Ua^i(t_2/)]$, then $\mathit{Ateb}(t) =
    (\lambda \Fliftcons{i}{\mathit{Ateb}(t_1)}) \Fshift{i}{\mathit{Ateb}(t_2)}$. By induction
    hypothesis, there exists $u_1$ and $u_2$ such that $\mathit{Ateb}(t_1)
    \ra_B^* u_1$ and $\mathit{Ateb}(t_2) \ra_B^* u_2$ with $u_1\rel t_1$ and $u_2\rel
    t_2$. We take $u'=(\lambda \Fliftcons{i}{u_1})
    \Fshift{i}{u_2}$ for which it is clear that
    $\mathit{Ateb}(t)\ra_B^*u'$. We have $u'\ra_B
    \Fliftcons{i}{u_1}[\Fshift{i}{u_2}/]$, we take this last term as $u$
    and we check that $u\rel t$, that is
    $\ol{u}=\ol{t}$ and $u\skelin t$. This last condition is trivial
    since $u_1\skelin t_1$ and $u_2\skelin t_2$. We calculate $\ol{u} =
    \ol{\Fliftcons{i}{u_1}[\Fshift{i}{u_2}/]}$, which is equal to
    $\Fliftcons{i}{\ol{u_1}}[\Fshift{i}{\ol{u_2}}/]$ by
    Lemmas~\ref{l:commute-ol-flc} and~\ref{l:commute-ol-fs}. $\ol{t} = \ol{t_1[\Ua^i(t_2/)]} =
    \Fliftcons{i}{\ol{t_1}}[\Fshift{i}{\ol{t_2}}/]$, and we conclude since
    $\ol{u_1}=\ol{t_1}$ and $\ol{u_2}=\ol{t_2}$.
  \end{itemize}
\end{proof}

\subsubsection{Simulation Lemmas}

We need several lemmas in order to prove the simulation of reductions of $\lambda\upsilon$.
We separate the reduction rules in two subset: we call $R_1$ the set containing the rule $B$
alone and $R_2$ the set containing all the other rules. Of course, $R_2$ is strongly normalizing
(see~\cite{Les96}). We want to establish the following diagrams:

$$
\begin{array}{ccc}
  t & \ra_B & t' \\
  \relsim & & \relsim \\
  u & \ra_{\lambda\upsilon}^+ & u'
\end{array}
\bigspace
\begin{array}{ccc}
  t & \ra_{R_2} & t' \\
  \relsim & & \relsim \\
  u & \ra_{\lambda\upsilon}^* & u'
\end{array}
$$

We look first at the simulation of $B$, then at that of the other.

\begin{lemma}\label{l:simul-B}
  For all $t\ra_B t'$, for all $u\rel t$ there exists $u'$ such that
  $u\ra_B u'$ and $u'\rel t'$.

  $$
  \begin{array}{ccc}
    t & \ra_B & t' \\
    \relsim & & \relsim \\
    u & \ra_B & u'
  \end{array}
  $$
\end{lemma}

\begin{proof}
  Let $t=(\lambda v) w$ and
  $(\lambda v) w\ra_B v[w/]$, every terms
  $u\rel t$ are of the form $(\lambda v') w'$ with $v'\rel v$ and
  $w'\rel w$, we can the reduce $(\lambda v') w'\ra_B v'[w'/]$
  and the conclusion follows immediately.
\end{proof}

\begin{lemma}\label{l:simul-R_2}
  For all $t\ra_{R_2} t'$, for all $u\rel t$ there exists $u'$ such that
  $u\ra^*_{\lambda\upsilon} u'$ and $u'\rel t'$.

  $$
  \begin{array}{ccc}
    t & \ra_{R_2} & t' \\
    \relsim & & \relsim \\
    u & \ra_{\lambda\upsilon}^* & u'
  \end{array}
  $$
\end{lemma}

\begin{proof}
  By case on the rule of $R_2$.
  \begin{itemize}
  \item $\mathit{FVar}$: $1[v/]\ra v$. Every terms
    $u\rel 1[v/]$ are of the form $1[v'/]$ with $v'\rel v$ and
    $1[v'/]\ra_{\mathit{FVar}} v'$.

    \medskip
  \item $\mathit{RVar}$: $n+1[v/]\ra n$. Every terms $u\rel
    n+1[v/]$ are of the form
    $n+1[v'/]$ with $v'\rel v$ and $n+1[v'/]\ra_{\mathit{RVar}} n$.

    \medskip
  \item $\mathit{App}$: $t=(v\ w)[s]\ra (v[s])\ (w[s])=t'$. We proceed by case on the form
    of$s$.
    \begin{itemize}
    \item if $s=\Ua^i(\ua)$ then the terms $u\rel (v\ w)[\Ua^i(\ua)]$
      might have two distinct forms:

      \begin{itemize}
      \item either $u=(v'\ w')[\Ua^j(\ua)]$ with $v'\rel v$, $w'\rel w$,
        $j\leq i$ and $\ol{u}=\ol{t}$, that is:

        $$
        \begin{array}{ccc}
          \ol{(v'\ w')[\Ua^j(\ua)]} &=& \ol{(v\ w)[\Ua^i(\ua)]} \\
          = & & = \\
          \Fliftshift{j}{\ol{v'}\ \ol{w'}} &=& \Fliftshift{i}{\ol{v}\ \ol{w}} \\
          = & & = \\
          \Fliftshift{j}{\ol{v'}}\ \Fliftshift{j}{\ol{w'}} & = &
          \Fliftshift{i}{\ol{v}}\ \Fliftshift{i}{\ol{w}} \\
        \end{array}
        $$

        which implies $\Fliftshift{j}{\ol{v'}}=\Fliftshift{i}{\ol{v}}$
        and $\Fliftshift{j}{\ol{w'}}=\Fliftshift{i}{\ol{w}}$. In that case,
        $(v'\ w')[\Ua^j(\ua~)~]\ra_{App}(v'[\Ua^j(\ua)])\ (w'[\Ua^j(\ua)])$
        and we can easily conclude with
        $(v'[\Ua^j(\ua)])\ (w'[\Ua^j(\ua)])\rel(v[\Ua^i(\ua)])\ (w[\Ua^i(\ua)])$.

      \item either $u=(v'\ w')$ with $v'\rel v$, $w'\rel w$, and $\ol{u}=\ol{t}$,
        that is:

        $$
        \begin{array}{ccc}
          \ol{v'\ w'} &=& \ol{(v\ w)[\Ua^i(\ua)]} \\
          = & & = \\
          \ol{v'}\ \ol{w'} &=& \Fliftshift{i}{\ol{v}\ \ol{w}} \\
          = & & = \\
          \ol{v'}\ \ol{w'} & = &
          \Fliftshift{i}{\ol{v}}\ \Fliftshift{i}{\ol{w}} \\
        \end{array}
        $$

        which implies $\ol{v'}=\Fliftshift{i}{\ol{v}}$
        and $\ol{w'}=\Fliftshift{i}{\ol{w}}$. In that case,
        $(v'\ w')$ can't be reduce and we can conclude with
        $(v'\ w')\rel(v[\Ua^i(\ua)])\ (w[\Ua^i(\ua)])$.
      \end{itemize}

    \item if $s=\Ua^i(r/)$ then all terms $u\rel(v\ w)[\Ua^i(r/)]$
      are of the form $(v'\ w')[\Ua^j(r'/)]$ with $v'\rel v$, $w'\rel w$,
      $r'\rel r$, $j\leq i$ and $\ol{u}=\ol{t}$, that is:

      $$
      \begin{array}{ccc}
        \ol{(v'\ w')[\Ua^j(r'/)]} &=& \ol{(v\ w)[\Ua^i(r/)]} \\
        = & & = \\
        \Fliftcons{j}{\ol{v'}\ \ol{w'}}[\Fshift{j}{\ol{r'}}/] &=&
        \Fliftcons{i}{\ol{v}\ \ol{w}}[\Fshift{i}{\ol{r}}/] \\
        = & & = \\
        (\Fliftcons{j}{\ol{v'}}\ \Fliftcons{j}{\ol{w'}})[\Fshift{j}{\ol{r'}}/] &=&
        (\Fliftcons{i}{\ol{v}}\ \Fliftcons{i}{\ol{w}})[\Fshift{i}{\ol{r}}/] \\
      \end{array}
      $$

      which implies $\Fliftshift{j}{\ol{v'}}=\Fliftshift{i}{\ol{v}}$,
      $\Fliftshift{j}{\ol{w'}}=\Fliftshift{i}{\ol{w}}$ and
      $\Fshift{j}{\ol{r'}}=\Fshift{i}{\ol{r}}$. In that case,
      $(v'\ w')[\Ua^j(r'/)]\ra_{\mathit{App}}(v'[\Ua^j(r'/)])\ (w'[\Ua^j(r'/)])$
      and we can easily conclude with
      $(v'[\Ua^j(r'/)])\ (w'[\Ua^j(r'/)])\rel(v[\Ua^i(r/)])\ (w[\Ua^i(r/)])$.
    \end{itemize}

    \medskip
  \item $\mathit{Lambda}$: $t=(\lambda v)[s]\ra\lambda(v[\Ua(s)])=t'$. We proceed by case
    on the form of $s$.
    \begin{itemize}
    \item if $s=\Ua^i(\ua)$ then the terms
      $u\rel(\lambda v)[\Ua^i(\ua)]$ might have two distinct forms:

      \begin{itemize}
      \item either $u=(\lambda v')[\Ua^j(\ua)]$ with $v'\rel v$,
        $j\leq i$ and $\ol{u}=\ol{t}$, that is :

        $$
        \begin{array}{ccc}
          \ol{(\lambda v')[\Ua^j(\ua)]} &=& \ol{(\lambda v)[\Ua^i(\ua)]} \\
          = & & = \\
          \Fliftshift{j}{\ol{\lambda v'}} &=& \Fliftshift{i}{\ol{\lambda v}} \\
          = & & = \\
          \lambda \Fliftshift{j+1}{\ol{v'}} & = &
          \lambda \Fliftshift{i+1}{\ol{v}} \\
        \end{array}
        $$

        which implies $\Fliftshift{j+1}{\ol{v'}}=\Fliftshift{i+1}{\ol{v}}$.
        In that case,
        $(\lambda v')[\Ua^j(\ua)]\ra_{\mathit{Lambda}}\lambda(v'[\Ua^{j+1}(\ua)])$
        and we can easily conclude with
        $\lambda(v'[\Ua^{j+1}(\ua)])\rel\lambda(v[\Ua^{i+1}(\ua)])$.

      \item either $u=\lambda v'$ with $v'\rel v$, and $\ol{u}=\ol{t}$
        that is:

        $$
        \begin{array}{ccc}
          \ol{\lambda v'} &=& \ol{(\lambda v)[\Ua^i(\ua)]} \\
          = & & = \\
          \lambda\ol{v'} &=& \Fliftshift{i}{\ol{\lambda v}} \\
          = & & = \\
          \lambda\ol{v'} & = &
          \lambda\Fliftshift{i+1}{\ol{v}} \\
        \end{array}
        $$

        which implies $\ol{v'}=\Fliftshift{i+1}{\ol{v}}$. In that case,
        $\lambda v'$ can't be reduced and we can conclude with
        $\lambda v'\rel\lambda(v[\Ua^{i+1}(\ua)])$.
      \end{itemize}

    \item if $s=\Ua^i(r/)$ then all the terms
      $u\rel(\lambda v)[\Ua^i(r/)]$ are of the form
      $(\lambda v')[\Ua^j(r'/)]$ with $v'\rel v$, $r\rel r'$,
      $j\leq i$ and $\ol{u}=\ol{t}$, that is:

      $$
      \begin{array}{ccc}
        \ol{(\lambda v')[\Ua^j(r'/)]} &=& \ol{(\lambda v)[\Ua^i(r/)]} \\
        = & & = \\
        \Fliftcons{j}{\ol{\lambda v'}}[\Fshift{j}{\ol{r'}}/] &=&
        \Fliftcons{i}{\ol{\lambda v}}[\Fshift{i}{\ol{r}}/] \\
        = & & = \\
        \lambda\Fliftcons{j+1}{\ol{v'}}[\Fshift{j}{\ol{r'}}/] & = &
        \lambda\Fliftcons{i+1}{\ol{v}}[\Fshift{i}{\ol{r}}/] \\
      \end{array}
      $$

      which implies $\Fliftshift{j+1}{\ol{v'}}=\Fliftshift{i+1}{\ol{v}}$
      and $\Fshift{j}{\ol{r'}}=\Fshift{i}{\ol{r}}$.
      In that case,
      $(\lambda v')[\Ua^j(r'/)]\ra_{\mathit{Lambda}}\lambda(v'[\Ua^{j+1}(r'/)])$
      and we can easily conclude with
      $\lambda(v'[\Ua^{j+1}(r'/)])\rel\lambda(v[\Ua^{i+1}(r/)])$
      due to Property~\ref{p:funshift+1-lu}.
    \end{itemize}

    \medskip
  \item $\mathit{VarShift}$: $n[\ua]\ra n+1$. The only two terms
    $u\rel n[\ua]$ are $n[\ua]$ and $n+1$, we can then conclude with
    possibly a reduction step using $\mathit{VarShift}$.

    \medskip
  \item $\mathit{FVarLift}$: $t=1[\Ua(s)] \ra_{\mathit{FVarLift}} 1=t'$. We proceed by
    case on the form of $s$.
    \begin{itemize}
    \item if $s=\Ua^i(\ua)$ then the terms
      $u\rel 1[\Ua(\Ua^i(\ua))]$ might have two distinct forms:

      \begin{itemize}
      \item either $u=1[\Ua(\Ua^j(\ua))]$ with
        $j\leq i$ and $\ol{u}=\ol{t}$. We then have
        $1[\Ua(\Ua^j(\ua))]\ra_{FVarLift} 1$ and we easily conclude.

      \item either $u=1$ with $\ol{u}=\ol{t}$ and we easily conclude.
      \end{itemize}

    \item if $s=\Ua^i(r/)$ then all the terms
      $u\rel 1[\Ua(\Ua^i(r/))]$ are of the form
      $1[\Ua(\Ua^j(r'/))]$ with $r'\rel r$,
      $j\leq i$ and $\ol{u}=\ol{t}$. We then have
      $1[\Ua(\Ua^j(r'/))]\ra_{\mathit{FVarLift}} 1$ and we can conclude.
    \end{itemize}

    \medskip
  \item $\mathit{RVarLift}$: $t=n+1[\Ua(s)]\ra n[s][\ua]=t'$. We proceed by
    case on the form of $s$.
    \begin{itemize}
    \item if $s=\Ua^i(\ua)$ then the terms
      $u\rel n+1[\Ua(\Ua^i(\ua))]$ might have two distinct forms:

      \begin{itemize}
      \item either $u=n'[\Ua^j(\ua)]$ with
        $j\leq i+1$ and $\ol{u}=\ol{t}$, that is:

        $$
        \begin{array}{ccc}
          \ol{n'+1[\Ua(\Ua^j(\ua))]} &=& \ol{n+1[\Ua(\Ua^i(\ua))]} \\
          = & & = \\
          \Fliftshift{j+1}{n'} &=& \Fliftshift{i+1}{n+1} \\
        \end{array}
        $$

        We deduce from this equality that $n'$ can't be smaller than $n$
        and that it must then be grater than $1$.
        In that case,
        $n'[\Ua(\Ua^j(\ua))]\ra_{\mathit{RVarLift}} n'-1[\Ua^j(\ua)][\ua]$
        and we must check that

        $$
        \begin{array}{ccc}
          \ol{n[\Ua^i(\ua)][\ua]} & = & \ol{n'-1[\Ua^j(\ua)][\ua]} \\
          = & & = \\
          \Fshift{1}{\Fliftshift{i}{n}} & = &
          \Fshift{1}{\Fliftshift{j}{n'-1}} \\
        \end{array}
        $$

        By Property~\ref{p:funshift-comp-lu}, we have
        $\Fshift{1}{\Fliftshift{i}{n}}=\Fliftshift{i+1}{n+1}$ and
        $\Fshift{1}{\Fliftshift{j}{n'-1}}=\Fliftshift{j+1}{n'}$, and
        we can conclude with
        $n'-1[\Ua^j(\ua)][\ua]\rel n[\Ua^i(\ua)][\ua]$.

      \item either $u=n'$, and $\ol{u}=\ol{t}$,
        that is:

        $$
        \begin{array}{ccc}
          \ol{n'} &=& \ol{n+1[\Ua(\Ua^i(\ua))]} \\
          = & & = \\
          n' &=& \Fliftshift{i+1}{n+1} \\
        \end{array}
        $$

        In that case, $u$ can't be reduced and we must check that

        $$
        \begin{array}{ccc}
          \ol{n[\Ua^i(\ua)][\ua]} & = & \ol{n'} \\
          = & & = \\
          \Fshift{1}{\Fliftshift{i}{n}} & = & n' \\
        \end{array}
        $$

        We conclude with $n'\rel n[\Ua^i(\ua)][\ua]$ due to
        Property~\ref{p:funshift-comp-lu}.
      \end{itemize}

    \item if $s=\Ua^i(r/)$ then all the terms
      $u\rel n+1[\Ua(\Ua^i(r/))]$ are of the form
      $n'[\Ua^j(r'/)]$ with $r\rel r'$,
      $j\leq i$ and $\ol{u}=\ol{t}$, that is:

      $$
      \begin{array}{ccc}
        \ol{n'[\Ua^j(r'/)]} &=& \ol{n+1[\Ua(\Ua^i(r/))]} \\
        = & & = \\
        \Fliftcons{j}{n'}[\Fshift{j}{\ol{r'}}/] & = &
        \Fliftcons{i+1}{n+1}[\Fshift{i+1}{\ol{r}}/] \\
      \end{array}
      $$

      which implies $\Fliftcons{j}{n'}=\Fliftcons{i+1}{n+1}$
      and $\Fshift{j}{\ol{r'}}=\Fshift{i+1}{\ol{r}}$.
      There are two distinct cases according to the value of $j$.

      \begin{itemize}
      \item $j=0$: then we have $\Fliftcons{0}{n'}=n'=\Fliftcons{i+1}{n+1}$,
        $\Fshift{0}{\ol{r'}}=\ol{r'}=\Fshift{i+1}{\ol{r}}$ and we must check that

        $$
        \begin{array}{ccc}
          \ol{n[\Ua^i(r/)][\ua]} & = & \ol{n'[r'/]} \\
          = & & \\
          \Fshift{1}{\Fliftcons{i}{n}[\Fshift{i}{r}/]} & & \\
          = & & \\
          \Fliftshift{0}{\Fliftcons{i}{n}[\Fshift{i}{r}/]} & & = \\
          = & & \\
          \Fliftshift{1}{\Fliftcons{i}{n}}[\Fliftshift{0}{\Fshift{i}{r}}/] & & \\
          = & & \\
          \Fliftshift{1}{\Fliftcons{i}{n}}[\Fshift{i+1}{r}/] & = & n'[\ol{r'}] \\
        \end{array}
        $$

        We can conclude with
        $\Fliftshift{1}{\Fliftcons{i}{n}}=\Fliftcons{i+1}{n+1}n'$
        due to Property~\ref{p:funcons-comp-lu}.

      \item $j>0$: $n'[\Ua^j(r'/)]$ reduces to
        $n'[\Ua^{j-1}(r'/)][\ua]$ and we must check that

        $$
        \begin{array}{ccc}
          \ol{n[\Ua^i(r/)][\ua]} & = & \ol{n'[\Ua^{j-1}(r'/)][\ua]} \\
          = & & = \\
          \Fshift{1}{\Fliftcons{i}{n}[\Fshift{i}{r}/]} & &
          \Fshift{1}{\Fliftcons{j-1}{n'}[\Fshift{j-1}{r'}/]} \\
          = & & = \\
          \Fliftshift{0}{\Fliftcons{i}{n}[\Fshift{i}{r}/]} & &
          \Fliftshift{0}{\Fliftcons{j-1}{n'}[\Fshift{j-1}{r'}/]} \\
          = & & = \\
          \Fliftshift{1}{\Fliftcons{i}{n}}[\Fliftshift{0}{\Fshift{i}{r}}/] & &
          \Fliftshift{1}{\Fliftcons{j-1}{n'}}[\Fliftshift{0}{\Fshift{j-1}{r'}}/] \\
          = & & = \\
          \Fliftshift{1}{\Fliftcons{i}{n}}[\Fshift{i+1}{r}/] & = &
          \Fliftshift{1}{\Fliftcons{j-1}{n'}}[\Fshift{j}{r'}/] \\
        \end{array}
        $$

        and we directly conclude with the help of Property~\ref{p:funcons-comp-lu}.

      \end{itemize}
    \end{itemize}
  \end{itemize}
\end{proof}

\subsubsection{Simulation}

The function $\mathit{Ateb}$ and the relation $\rel$ satisfy the hypothesis of
Theorem~\ref{th:simul}. We can then apply it and get the desired conclusion.

\begin{corollary}
  Since $\lambda\upsilon$-calculus enjoys PSN~\cite{Les96} and simply-typed
  $\lambda$-calculus enjoys SN~\cite{Kri90} (which is easily extended to
  $\lambda$-calculus with De Bruijn indices), we have that simply-typed
  $\lambda\upsilon$-calculus enjoys SN.
\end{corollary}

\section{$\lwsn$-calculus}

In~\cite{CosKesPol03} a named version of $\lws$ was proposed. In current
work, we developed a new version of this calculus : $\lwsn$. We already have a
SN proof for this calculus, almost similar to the original one, and this
technique can be applied, using the direct proof. We cannot conclude to SN by this
way, since PSN has not yet been shown (see~\cite{PolPhD04}).

\subsection{Definition}

Terms of $\lwsn$-calculus are given by the following grammar:
 
$$
t ::= x\ |\ (t\ t)\ |\ \lam x.t\ |\ t[x,t,\Gamma,\Gamma]\ |\ \Gamma t
$$

where $\Gamma$ is a set of variable. A version of the reduction rules is presented
Fig.~\ref{f:def-lln}.

\begin{figure}[htb]
  $$
  \begin{array}{crcll}
    (b) & (\Delta(\lam x.t)) (\Gamma u) &
    \ra & t[x,u,\Gamma,\Delta] & \\ \\
    (a) & (t\ u)[x,v,\Gamma,\Delta] &
    \ra & (t[x,v,\Gamma,\Delta]\ u[x,v,\Gamma,\Delta]) & \\
    (e_1) & (\Lambda t)[x,u,\Gamma,\Delta] &
    \ra & (\Delta \cup (\Lambda\sm \{x\}))t & x\in\Lambda\sm\Gamma \\
    (n_1) & y[x,t,\Gamma,\Delta] & \ra & \Delta y & x\not= y\mbox{ or }y\in\Gamma\\
    (n_2) & x[x,t,\Gamma,\Delta] & \ra & \Gamma t & \\
    (c_1) & t[y,u,\Lambda,\Phi][x,v,\Gamma,\Delta] &
    \ra & & x\in\Phi\sm\Gamma\mbox{ and }x\not\in\Lambda\sm\Gamma\\
    & \multicolumn{3}{c}{
      t[y,u[x,v,\Gamma\sm\Lambda,\Delta\cup(\Lambda\sm\Gamma)],
        \Lambda\cap\Gamma,\Delta \cup (\Phi \sm \{x\})]} \\
    (c_2) & t[y,u,\Lambda,\Phi][x,v,\Gamma,\Delta] &
    \ra & t[x,v,(\Gamma\sm \Phi)\cup\{y\},\Delta\cup(\Phi\sm\Gamma)] \\
    & & &
    \phantom{t}[y,u[x,v,\Gamma\sm\Lambda,\Delta\cup(\Lambda\sm\Gamma)],
      \Lambda\cap\Gamma,\Gamma \cap\Phi] & x\not\in\Phi\sm\Gamma\mbox{
    and }x\not\in\Lambda\sm\Gamma \\
    (c_3) & t[y,u,\Lambda,\Phi][x,v,\Gamma,\Delta] &
    \ra & t[y,u,(\Lambda\sm\{x\})\cup\Delta,(\Phi\sm\{x\})\cup\Delta] &
    x\in\Phi\sm\Gamma\mbox{ and }x\in\Lambda\sm\Gamma \\ \\
    (f) & (\lam y.t)[x,u,\Gamma,\Delta] &
    \sim & \lam y.t[x,u,\Gamma\cup\{y\},\Delta] & \\
    (e_2) & (\Lambda t)[x,u,\Gamma,\Delta] &
    \sim & (\Gamma \cap \Lambda)t
    [x,u,\Gamma\sm\Lambda,\Delta\cup(\Lambda\sm\Gamma)] &
    x\not\in\Lambda\sm\Gamma \\
    (d) & \Gamma\Delta t & \sim & (\Gamma\cup \Delta)t & \\
    (\es) & \es t & \sim & t & \\
    (c_4) & t[y,u,\Lambda,\Phi][x,v,\Gamma,\Delta] &
    \sim & & x\in\Lambda\sm\Gamma\mbox{ and }x\not\in\Phi\sm\Gamma \\
    & \multicolumn{3}{c}{t[x,v,(\Gamma\sm\Phi)\cup\{y\},
        \Delta\cup(\Phi\sm\Gamma)]
      [y,u,\Delta\cup(\Lambda\sm\{x\}),\Gamma\cap\Phi]} & \\
  \end{array}
  $$
  \caption{Reduction rules of the $\lwsn$-calculus}
  \label{f:def-lln}
\end{figure}

Typing rules are given Fig.~\ref{f:lln-typage}.
 
\begin{figure}[htb]
\[
\begin{array}{l@{\hspace{1.5cm}}l}
\infer[Ax]{x:A \vdash x:A}{} &
\infer[Weak]{\Gamma\vdash \Delta t:A}
            {\Gamma\sm\Delta \vdash t:A & \Delta\subset\Gamma} \\
& \\
\infer[App]{\Gamma \vdash (t\ u):A}
           {\Gamma \vdash t:B \> A & \Gamma \vdash u:B} &
\infer[Lamb]{\Gamma \vdash \lambda x.t:B \> A}
              {\Gamma,x:A \vdash t:B}\\
\end{array}
\]
\[
\infer[Sub]{\Pi \vdash t[x,u,\Gamma,\Delta]:B}
             {\Pi\sm\Gamma \vdash u:A & \Pi\sm\Delta,x:A \vdash t:B
               & (\Gamma\cup\Delta)\subset\Pi}
\]
\caption{Typing rules of the $\lwsn$-calculus}
\label{f:lln-typage}
\end{figure}                                                                    

\subsection{Strong Normalization proof}

We define the $\mathit{Ateb}$ function as follows:

\begin{center}
  \begin{tabular}{lll}
    $\mathit{Ateb}(x)$ & $=$ & $x$ \\
    $\mathit{Ateb}(t\ u)$ & $=$ & $\mathit{Ateb}(t)\ \mathit{Ateb}(u)$ \\
    $\mathit{Ateb}(\lambda x.t)$ & $=$ & $\lambda x.\mathit{Ateb}(t)$ \\
    $\mathit{Ateb}(\Gamma t)$ & $=$ & $\Gamma \mathit{Ateb}(t)$ \\
    $\mathit{Ateb}(t[x,u,\Gamma,\Delta])$ & $=$ & $(\Delta(\lambda x.\mathit{Ateb}(t)))\ (\Gamma \mathit{Ateb}(u))$ \\
  \end{tabular}
\end{center}

\begin{remark}
  The $\mathit{Ateb}$ function sends $\lwsn$-terms to a $\lambda$-calculus with explicit weakening.
\end{remark}

As for the $\lambda{\tt x}$-calculus, the $\mathit{Ateb}$ function performs exactly the reverse
reduction of the rule $b$. It is then obvious that if $t'=\mathit{Ateb}(t)$ then
$t'\ra_{b}^*t$ and that $\mathit{Ateb}(t)$ does not contain any substitution. We must check that
the term we get is typeable.

\begin{lemma}
  $$\Gamma\vdash t:A \;\;\Ra\;\; \Gamma\vdash \mathit{Ateb}(t):A$$
\end{lemma}

\begin{proof}
  By induction on the typing derivation of $t$. The only interesting case is that of substitution.
  We have $t=u[x,v,\Gamma,\Delta]$ and

  $$
  \infer
    {\Pi\vdash u[x,v,\Gamma,\Delta]:A}
    {\Pi\sm\Gamma\vdash v:B & \Pi\sm\Delta,x:B\vdash u:A}
  $$

  By induction hypothesis, we have $\Pi\sm\Delta,x:B\vdash \mathit{Ateb}(u):A$ and
  $\Pi\sm\Gamma\vdash \mathit{Ateb}(v):B$. We can type $\mathit{Ateb}(t)=(\Delta(\lambda
  x.\mathit{Ateb}(u)))\ \Gamma\mathit{Ateb}(v)$ as follows

  $$
  \infer
    {\Pi\vdash (\Delta(\lambda x.\mathit{Ateb}(u)))\ \Gamma\mathit{Ateb}(v):A}
    {\infer
      {\Pi\vdash \Delta(\lambda x.\mathit{Ateb}(u)):B\ra A}
      {\infer
        {\Pi\sm\Delta\vdash \lambda x.\mathit{Ateb}(u):B\ra A}
        {\Pi\sm\Delta,x:B\vdash \mathit{Ateb}(u):A}} &
      \infer
      {\Pi\vdash\Gamma\mathit{Ateb}(v):B}
      {\Pi\sm\Gamma\vdash \mathit{Ateb}(v):B}}
  $$
\end{proof}

We can directly apply Theorem~\ref{th:direct}. Nevertheless, we cannot get any conclusion since
PSN has not yet been shown for this calculus.

\section{$\lws$-calculus}

We deal here with the calculus with De Bruijn indices, and difficulties will arise due to them.
More precisely, we won't be able to deal with the typing environment as we did for the
$\lambda\upsilon$-calculus. The presence of explicit weakening forbid us to rearrange the typing
environment as far as we would do. Here follows the reduction rules (Fig.~\ref{f:reduction-ll})
and typing rules of the $\lws$-calculus (Fig.~\ref{f:typage-ll}) where $|\Gamma| = i$ and $|\Delta| =j$.

\begin{figure}[htb]
\[\begin{array}{crcll}
b_1 & ({\lam}t u)        & \ra & [0/u,0]t                     & \\
b_2 & (\lab{k}{\lam}t u) & \ra & [0/u,k]t                     & \\
f   & [i/u,j]{\lam}t     & \ra & \lam[i+1/u,j]t               & \\
a   & [i/u,j](t\ v)      & \ra & (([i/u,j]t)\ ([i/u,j]v))     & \\
e_1 & [i/u,j]\lab{k}t    & \ra & \lab{j+k-1}t                 & si\ i<k \\
e_2 & [i/u,j]\lab{k}t    & \ra & \lab{k}[i-k/u,j]t            & si\ i\ge k \\
n_1 & [i/u,j]k           & \ra & k                            & si\ i>k \\
n_2 & [i/u,j]i           & \ra & \lab{i}u                     & \\
n_3 & [i/u,j]k           & \ra & j+k-1                        & si\ i<k \\
c_1 & [i/u,j][k/v,l]t    & \ra & [k/[i-k/u,j]v,j+l-1]t        & si\ k\le i<k+l \\
c_2 & [i/u,j][k/v,l]t    & \ra & [k/[i-k/u,j]v,l][i-l+1/u,j]t & si\ i\ge k+l \\
d   & \lab{i}\lab{j}t    & \ra & \lab{i+j}t                   & \\
\end{array}\]
\caption{Reduction rules}
 \label{f:reduction-ll}
\end{figure}

\begin{figure}[htb]
\[ \infer[Axiom]{\Gamma, A, \Delta  \vdash  i:A}{} \] 
\[ \begin{array}{l@{\hspace{2cm}}l}
\infer[Lambda]{\Gamma \vdash  \lambda t:B \> C}
        {B, \Gamma \vdash  t:C} &
\infer[App]{\Gamma  \vdash  (t  u):A}
        {\Gamma  \vdash  t:B\> A & \Gamma \vdash  u:B}\\\\
\infer[Subst]{\Gamma, \Delta, \Pi \vdash  [i/u,j]t:B}
        {\Delta, \Pi \vdash u:A & \Gamma, A, \Pi \vdash t:B}&
\infer[Weak]{\Gamma, \Delta \vdash \lab{i}t:B}
        {\Delta \vdash t:B}\\\\
\end{array} \] 
\caption{Typing rules}
 \label{f:typage-ll}
\end{figure}

At the time of writing, we don't know if it would be possible to apply
our technique to this calculus.

\section{$\lsigma$-calculus}

The $\lsigma$-calculus~\cite{AbaCarCurLev91} is a calculus with De Bruijn indices
and multiple substitutions, adding difficulties over those already there for
the $\lambda\upsilon$-calculus. Our application here is only an exercise since
this calculus does not enjoy PSN. Nevertheless, it reduces the question of SN
to that of PSN, {\em i.e.} if PSN is shown, here already follows a correct proof
of SN.

\subsection{Definition}

Terms of the $\lsigma$-calculus are given by the following grammar:

$$
\begin{array}{l}
  t ::= 1\ |\ (t\ t)\ |\ \lambda t\ |\ t[s] \\
  s ::= id\ |\ \ua\ |\ t\cdot s\ |\ s\circ s \\
\end{array}
$$

As usual, we will add infinitely many integer constants $2,3,...,n$ with the convention:
$n=\indice{n-1}$. As usually, we will consider that any term $n$ does not contain substitutions.

Here follows the reduction rules:

$$
\begin{array}{lll}
  (\lambda t) u & \ra_{B} & t[u\cdot id] \\ \\
  (t\ u)[s] & \ra_{App} & (t[s])\ (u[s]) \\
  (\lambda t)[s] & \ra_{Lambda} & \lambda (t[1\cdot (s\circ\ua)]) \\
  1[id] & \ra_{VarId} & 1 \\
  1[t\cdot s] & \ra_{VarCons} & t \\
  t[s][s'] & \ra_{Clos} & t[s\circ s'] \\ \\
  id\circ s & \ra_{IdL} & s \\
  \ua\circ id & \ra_{ShiftId} & \ua \\
  \ua\circ(t\cdot s) & \ra_{ShiftCons} & s \\
  (t\cdot s)\circ s' & \ra_{Map} & t[s']\cdot(s\circ s') \\
  (s_1\circ s_2)\circ s_3 & \ra_{Ass} & s_1\circ(s_2\circ s_3) \\
\end{array}
$$

Here follows the typing rules:

\begin{center}
  \begin{tabular}{cc}
    \infer{A,\Gamma\vdash 1:A}{} &
    \infer
      {\Gamma\vdash t[s]:A}
      {\Gamma\vdash s\trg \Gamma' & \Gamma'\vdash t:A} \\ & \\
    \infer
      {\Gamma\vdash (t\ u):A}
      {\Gamma\vdash t:B\ra A & \Gamma\vdash u:B} &
    \infer
      {\Gamma\vdash \lambda t:B\ra A}
      {B,\Gamma\vdash t:A} \\ & \\
    \infer{\Gamma\vdash id\trg \Gamma}{} &
    \infer{A,\Gamma\vdash \ua\trg \Gamma}{} \\ & \\
    \infer
      {\Gamma\vdash t\cdot s\trg A,\Gamma'}
      {\Gamma\vdash t:A & \Gamma\vdash s\trg \Gamma'} &
    \infer
      {\Gamma\vdash s\circ s'\trg \Gamma'}
      {\Gamma\vdash s'\trg\Gamma'' & \Gamma''\vdash s\trg \Gamma'}
  \end{tabular}
\end{center}

We can give a derived rule for indices $n>1$, (with $n=|\Gamma|+1$ and
$\Gamma=C_1,...,C_{n-1}$) :

$$
\begin{array}{l@{\espace}r}
  \infer{\Gamma,A,\Delta\vdash n:A}{} &
  \infer
    {\Gamma,A,\Delta\vdash \indice{n-1}:A}
    {\Gamma,A,\Delta\vdash\ua\trg C_2,...,C_{n-1},A,\Delta &
     \infer
       {C_2,...,C_{n-1},A,\Delta\vdash
         \indice{n-2}:A}
       {\infer{\vdots}
         {C_{n-1},A,\Delta\vdash\ua\trg A,\Delta &
          \infer{A,\Delta\vdash 1:A}{}}}}
\end{array}
$$

The substitution back-pushing and some of the functions defined below were
strongly inspired by~\cite{CurRio91}.

\subsection{Towards strong normalization}

We proceed similarly to Section~\ref{sect:lu-sn}.

We define the $\mathit{Ateb}$ function as follows:

$$
\begin{array}{lll}
  \mathit{Ateb}(n) & = & n \\
  \mathit{Ateb}(t\ u) & = & Ateb(t)\ Ateb(u) \\
  \mathit{Ateb}(\lambda t) & = & \lambda Ateb(t) \\
  \mathit{Ateb}(t[id]) & = & Ateb(t) \\
  \mathit{Ateb}(t[\ua]) & = & \Upshift{0}{1}{Ateb(t)} \\
  \mathit{Ateb}(t[s\circ s']) & = & Ateb(t[s][s']) \\
  \mathit{Ateb}(t[t'\cdot s]) & = & Ateb((\lambda t)[s])\ Ateb(t') \\
\end{array}
$$

Where $\Upshift{i}{j}{t}$ is a function that we will define below. The goal of
this function is to anticipate the propagation of the substitution $[\ua]$ and
to perform early re-indexing. To understand its necessity, let us look at
the derivation of $t[\ua]$.

$$
\infer
  {B,\Gamma\vdash t[\ua]:A}
  {\infer{B,\Gamma\vdash\ua\trg\Gamma}{} &
   \Gamma\vdash t:A}
$$

\begin{example}
  For instance, if we suppose that for any $tt$ among $t$, $u$, $v$
  we have $tt=\mathit{Ateb}(tt)$, then we get
  $$
  \mathit{Ateb}((t[u\cdot id]\ v[1\cdot1\cdot 5\cdot\ua])[\ua]) \;\;\; = \;\;\;
  \Upshift{0}{1}{((\lambda t) u)\ 
    ((( \Upshift{0}{1}{\lambda\lambda\lambda v} 5) 1) 1)}
  $$
\end{example}

The calculus $\Upshift{0}{1}{t}$ will then increase by $1$ all the free
variables of $t$ in order to enable the typing of it in the environment $B,\Gamma$.
We can therefore state the property that this function must verify.

\begin{property}\label{p:fct-ls}
  For any term $t$ without substitutions we have: $\Gamma\vdash t:A \Ra
  B,\Gamma\vdash\Upshift{0}{1}{t}:A$.
\end{property}

It is obvious that for any $t$, $\mathit{Ateb}(t)$ does not contain any substitutions. We
can check that $\mathit{Ateb}(t)$ is typeable.

\begin{lemma}
  $$\Gamma\vdash t:A \;\;\Ra\;\; \Gamma\vdash\mathit{Ateb}(t):A$$
\end{lemma}

\begin{proof}
  By induction on $t$.

  \begin{itemize}
  \item $t=1$ and

    $$
    \infer{A,\Delta\vdash 1:A}{}
    $$

    We then have $\mathit{Ateb}(t)=1$ and the same typing derivation.

  \item $t=(u\ v)$ and

    $$
    \infer
      {\Gamma\vdash (u\ v):A}
      {\Gamma\vdash u:B\ra A & \Gamma\vdash v:B}
    $$

    By induction hypothesis, we have $\Gamma\vdash \mathit{Ateb}(u):B\ra A$ and
    $\Gamma\vdash \mathit{Ateb}(v):B$. We can type $\mathit{Ateb}(t)=\mathit{Ateb}(u)\
    \mathit{Ateb}(v)$ as follows

    $$
    \infer
      {\Gamma\vdash (\mathit{Ateb}(u)\ \mathit{Ateb}(v)):A}
      {\Gamma\vdash \mathit{Ateb}(u):B\ra A & \Gamma\vdash \mathit{Ateb}(v):B}
    $$

  \item $t=\lambda u$ and

    $$
    \infer
      {\Gamma\vdash \lambda u:B\ra A}
      {B,\Gamma\vdash u:A}
    $$

    By induction hypothesis, we have $B,\Gamma\vdash \mathit{Ateb}(u):A$. We
    can type $\mathit{Ateb}(t)=\lambda \mathit{Ateb}(u)$ as follows

    $$
    \infer
      {\Gamma\vdash \lambda \mathit{Ateb}(u):B\ra A}
      {B,\Gamma\vdash \mathit{Ateb}(u):A}
    $$

  \item $t=u[id]$ and

    $$
    \infer
      {\Gamma\vdash u[id]:A}
      {\Gamma\vdash id\trg\Gamma & \Gamma\vdash u:A}
    $$

    We directly conclude by induction hypothesis.

  \item $t=u[\ua]$ and

    $$
    \infer
      {B,\Gamma\vdash u[\ua]:A}
      {B,\Gamma\vdash \ua\trg\Gamma & \Gamma\vdash u:A}
    $$

    We conclude by induction hypothesis and by Property~\ref{p:fct-ls}.

  \item $t=u[v\cdot s]$ and

    $$
    \infer
      {\Gamma\vdash u[v\cdot s]:A}
      {\infer
        {\Gamma\vdash v\cdot s\trg B,\Gamma'}
        {\Gamma\vdash s\trg\Gamma' & \Gamma\vdash v:B} &
       B,\Gamma'\vdash u:A}
    $$

    By induction hypothesis, we have $\Gamma\vdash\mathit{Ateb}((\lambda u)[s]):B\ra
    A$ and $\Gamma\vdash\mathit{Ateb}(v):B$. We conclude with the following typing derivation

    $$
    \infer
      {\Gamma\vdash\mathit{Ateb}((\lambda u)[s])\ \mathit{Ateb}(v):A}
      {\Gamma\vdash\mathit{Ateb}((\lambda u)[s]):B\ra A &
       \Gamma\vdash\mathit{Ateb}(v):B}
    $$

  \item $t=u[s\circ s']$, and we conclude directly by induction hypothesis.
  \end{itemize}
\end{proof}

\subsubsection{Function definition}

The function $\Upshift{i}{j}{t}$ performs a re-indexing of the term $t$
as if we had propagated a substitution $[\ua]$. Since it deals only with terms
obtained from the $\mathit{Ateb}$ function, we might consider only substitution-free
terms. However, we will need later to use it with terms with substitutions, but without
$\ua$. When it is applied to a substitution it returns a pair composed by an integer and
a substitution, else it returns a term.

Here follows its complete definition:

$$
\begin{array}{llll}
  \Upshift{i}{j}{n} & = & n+j & \mbox{if } n>i \\
  \Upshift{i}{j}{n} & = & n & \mbox{if } n\leq i \\
  \Upshift{i}{j}{t\ u} & = & \Upshift{i}{j}{t}\
  \Upshift{i}{j}{u} & \\
  \Upshift{i}{j}{\lambda t} & = & \lambda\Upshift{i+1}{j}{t} & \\
  & & & \\
  \Upshift{i}{j}{t[s]} & = & \mbox{let } i',s' = \Upshift{i}{j}{s} & \\
  & & \mbox{in } \Upshift{i'}{j}{t}[s'] & \\
  \Upshift{i}{j}{id} & = & i,id & \\
  \Upshift{i}{j}{t\cdot s} & = & \mbox{let } i',s' =
  \Upshift{i}{j}{s} & \\
  & & \mbox{in } i'+1,\Upshift{i}{j}{t}\cdot s' & \\
  \Upshift{i}{j}{s_1\circ s_2} & = & \mbox{let } i_2',s_2' =
  \Upshift{i}{j}{s_2} & \\
  & & \mbox{and } i_1',s_1' = \Upshift{i_2'}{j}{s_1} & \\
  & & \mbox{in } i_1',s_1'\circ s_2' & \\
\end{array}
$$

The modification of the index $i$ (and the value of the integer part of
the pair) reflects the number of $\cdot$ we got through, each of them
acting like a $\lambda$.

Here follows the proof of Property~\ref{p:fct-ls}.

\begin{proof}
  We have to proof that, for any $t$ substitution-free,
  $\Gamma\vdash t:A \Ra B,\Gamma\vdash\Upshift{0}{1}{t}:A$. Actually we
  prove a more general result, namely $\Gamma,\Delta\vdash t:A \Ra
  \Gamma,B,\Delta\vdash\Upshift{i}{1}{t}:A$ where $i=|\Gamma|$. We proceed
  by induction on $t$.
  \begin{itemize}
  \item $t=n$ with $n\leq i$: $\Upshift{i}{1}{t}=n$. We have

    $$
    \infer{\Gamma_1,A,\Gamma_2,\Delta\vdash n:A}{}
    $$

    with $n=|\Gamma_1|+1$. We conclude with the following typing derivation

    $$
    \infer{\Gamma_1,A,\Gamma_2,B,\Delta\vdash n:A}{}
    $$

  \item $t=n$ with $n>i$ : $\Upshift{i}{1}{t}=n+1$. On a

    $$
    \infer{\Gamma,\Delta_1,A,\Delta_2\vdash n:A}{}
    $$

    With $n=|\Gamma|+|\Delta_1|+1$. We conclude with the following typing derivation

    $$
    \infer{\Gamma,B,\Delta_1,A,\Delta_2\vdash n+1:A}{}
    $$

  \item $t=(u\ v)$ : $\Upshift{i}{1}{t} = (\Upshift{i}{1}{u}\
    \Upshift{i}{1}{v})$. We conclude with twice the induction hypothesis.
  \item $t=\lambda u$ (with $A=C\ra D$) : $\Upshift{i}{1}{t} =
    \lambda\Upshift{i+1}{1}{u}$. We have

    $$
    \infer
      {\Gamma,\Delta\vdash \lambda u:C\ra D}
      {C,\Gamma,\Delta\vdash u:D}
    $$

    By induction hypothesis, we have $C,\Gamma,B,\Delta\vdash
    \Upshift{i+1}{1}{u}:D$, and we conclude with the following typing derivation

    $$
    \infer
      {\Gamma,B,\Delta\vdash \lambda\Upshift{i+1}{1}{u}:C\ra D}
      {C,\Gamma,B,\Delta\vdash \Upshift{i+1}{1}{u}:D}
    $$
  \end{itemize}
\end{proof}

Here follows a property used below.

\begin{property}\label{prop:comp-up}
  For all $t$, $i$, $j$, $l$, we have

  $$
  \Upshift{i}{j}{\Upshift{i}{l}{t}} \;\;\; = \;\;\; \Upshift{i}{j+l}{t}
  $$
\end{property}

\begin{proof}\label{proof:comp-up}
  By easy induction on $t$.
\end{proof}

\begin{example}
  We can apply this function to our example, giving
  $$
  \Upshift{0}{1}{((\lambda t) u)\ 
    (\lambda ((\lambda (\Upshift{0}{1}{\lambda v} w)) 1)) 1} \;\;\; =
  ((\lambda \Upshift{1}{1}{t}) \Upshift{0}{1}{u})\ 
  ((((\lambda\lambda\lambda\Upshift{3}{2}{v})6)2)2)
  $$
\end{example}

\subsubsection{Definition of the relation $\rel$}

The function $\mathit{Ateb}$ applied to a term $t$ returns a new term
$t'$ that usually cannot be reduce to $t$. Indeed, the $\ua$
disappears and the information they carried is already propagated in $t'$.
The reducts of $t'$ won't have those terms as it is shown in the following
example.

\begin{example}
  With our last example, we have

  $$
  \begin{array}{c}
    (\ub{(\lambda \Upshift{1}{1}{t}) \Upshift{0}{1}{u}})\ 
    ((((\lambda\lambda\lambda\Upshift{3}{2}{v})6)2)2) \\
    \ra^* \\
    \Upshift{1}{1}{t}[\Upshift{0}{1}{u}\cdot id]\ 
    \Upshift{3}{2}{v}[2\cdot 2\cdot 6\cdot id] \\
  \end{array}
  $$

  Remark that the re-indexing of the $\ua$ in the original term has correctly been
  propagated, the substitution $[1\cdot
  1\cdot 5\cdot \ua]$ has become $[2\cdot 2\cdot 6\cdot id]$.
\end{example}

We now have to simulate the reduction of the initial term by that of the
obtained term. We start naively with the following definition which will
appear to be inadequate. We will then present an adequate solution.

\bigskip

To perform the simulation, we define a new function $\ol{t}$ that
flattens all the re-indexing required in a term $t$ and deletes
the lonely substitutions $[id]$.

$$
\begin{array}{lll}
  \ol{n} & = & n \\
  \ol{t\ u} & = & \ol{t}\ \ol{u} \\
  \ol{\lambda t} & = & \lambda \ol{t} \\
  \ol{t[s]} & = & \mbox{let } n,s' = \ol{s} \mbox{ in} \\
  & & \Upshift{0}{n}{\ol{t}}[s'] \mbox{ if } s'\not=\es \\
  & & \Upshift{0}{n}{\ol{t}} \mbox{ else} \\
  \ol{\ua} & = & 1,\es \\
  \ol{id} & = & 0,\es \\
  \ol{t\cdot s} & = & \mbox{let } n,s' = \ol{s} \mbox{ in} \\
  & & n,\ol{t}[s'] \mbox{ if } s'\not=\es \\
  & & n,\ol{t}\cdot id \mbox{ else} \\
  \ol{s_1\circ s_2} & = & \mbox{let } n_1,s_1' = \ol{s_1} \\
  & & \mbox{and } n_2,s_2' = \ol{s_2} \mbox{ in} \\
  & & n_1+n_2,\es \mbox{ if } s_1'=s_2'=\es \\
  & & n_1+n_2,\Upshift{0}{n_2}{s_1'} \mbox{ if } s_2'=\es \\
  & & n_1+n_2,s_2' \mbox{ if } s_1'=\es \\
  & & n_1+n_2,\Upshift{0}{n_2}{s_1'}\circ s_2' \mbox{ else} \\
\end{array}
$$

The function $\ol{\cdot}$ commutes with $\Upshift{i}{j}{t}$,
as stated in the following lemma.

\begin{lemma}\label{l:commute-ol-ups}
  For all $i$, $j$ and $t$ (without $\ua$) we have

  $$
  \ol{\Upshift{i}{j}{t}} \;\;\;\; = \;\;\;\; \Upshift{i}{j}{\ol{t}}
  $$
\end{lemma}

\begin{proof}
  By induction on $t$.

  \begin{itemize}
  \item If $t=n$, then $\Upshift{i}{j}{t}=n'$, $\ol{n'}=n'$ and $\ol{n} = n$.
  \item All the remaining cases are easily proved by induction hypothesis.
  \end{itemize}
\end{proof}

\begin{example}
  Look at the final example term:
  $$
  \begin{array}{c}
    \ol{\Upshift{1}{1}{t}[\Upshift{0}{1}{u}\cdot id]\ 
    \Upshift{3}{2}{v}[2\cdot 2\cdot 6\cdot id]} \\
    = \\
    \Upshift{1}{1}{\ol{t}}[\Upshift{0}{1}{\ol{u}}\cdot id]\ 
    \Upshift{3}{2}{\ol{v}}[2\cdot 2\cdot 6\cdot id]
  \end{array}
  $$
  And the original term:
  $$
  \begin{array}{c}
    \ol{(t[u\cdot id]\ v[1\cdot1\cdot 5\cdot\ua])[\ua]} \\
    = \\
    \Upshift{0}{1}{\Upshift{0}{0}{\ol{t}}[\ol{u}\cdot id]\ 
      \Upshift{0}{1}{\ol{v}}[1\cdot1\cdot5\cdot id]}\\
    = \\
    \Upshift{1}{1}{\ol{t}}[\Upshift{0}{1}{\ol{u}}\cdot id]\ 
    \Upshift{3}{2}{\ol{v}}[2\cdot 2\cdot 6\cdot id]    
  \end{array}
  $$
\end{example}

We need an order relation on the skeleton of terms. We want that
$t\skelin t'$ if and only if $t'$ does contain $\ua$ and $[id]$ only
at the same place $t$ does. More formally, it can be defined as follows:

$$
\begin{array}{ccc}
  \mbox{pour tout } n\mbox{ and } m & & n\skelin m \\
  t\skelin t' \mbox{ and } u\skelin u' & \Ra & (t\ u)\skelin(t'\ u') \\
  t\skelin t' & \Ra & \lambda t\skelin\lambda t' \\
  t\skelin t' & \Ra & t\skelin t'[\ua] \\
  t\skelin t' & \Ra & t\skelin t'[id] \\
  t\skelin t' \mbox{ and } s\skelin s' & \Ra & t[s]\skelin t'[s'] \\
\end{array}
$$

$$
\begin{array}{ccc}
  & & \ua\skelin \ua \\
  & & id\skelin id \\
  & & id\skelin \ua \\
  t\skelin t' \mbox{ and } s\skelin s' & \Ra & t\cdot s\skelin t'\cdot s' \\
  s\skelin s' & \Ra & s \skelin s'\circ id \\
  s\skelin s' & \Ra & s \skelin id\circ s'  \\
  s\skelin s' & \Ra & s \skelin s'\circ \ua \\
  s\skelin s' & \Ra & s \skelin \ua\circ s'  \\
  s_1\skelin s_1' \mbox{ and } s_2\skelin s_2' & \Ra & s_1\circ
  s_2\skelin s_1'\circ s_2' \\
\end{array}
$$

\begin{example}
We have $t[t'\cdot id]\skelin t[id][t'\cdot\ua]$.
\end{example}

With this relation and the function $\ol{t}$, we can define a relation
to perform our simulation. We denote this relation $\rel$ and
we define it as follows:

$$
t\rel t' \iff \ol{t} = \ol{t'} \mbox{ and } t\skelin t'
$$

We remark that we always have $t\rel t$.

\bigskip

However, we cannot go further because this relation will not be adequate
to perform the simulation. Indeed, a problem arise to simulate the rule
$Abs$: $(\lambda t)[s]\ra \lambda(t[1\cdot (s\circ\ua)]$.
If $s=id$ (or $\ua$), then a term $u\rel(\lambda t)[id]$ can be $\lambda t$
which don't verify  $\lambda t\rel\lambda(t[1\cdot(s\circ\ua)]$.
We would like to extend our relation to this kind of $id$ (and similarly for $\ua$).
We start over again with a new relation $\rel$ that take this into account.

\bigskip

To solve the problem, we choose to identify terms having the same $\sigma$-normal form.
We call $\sigma$ the set of rules without $B$. The $\sigma$-normal form
of a term is given by the transitive closure of $\sigma$. We now that such
a normal form exists since $\sigma$ is strongly normalizing
(see~\cite{AbaCarCurLev91}). We denote $\sigma(t)$ the $\sigma$-normal
form of $t$.

\bigskip

We first define a notion of {\em redexability} of terms. The idea is to find
all the ``bad'' terms, that is those that can give rise to $B$-redices.

\begin{definition}\label{def:redexabilite}
  We say that a term is potentially redexable (denoted $PR(t)$)
  if it contains application or $\lambda$ at some node.
\end{definition}

We define then the relation $\skelin$ that will ensure that if
$u\skelin t$ then $u$ has the same redexability than $t$.

$$
\begin{array}{cccc}
  \mbox{pour tout } n\mbox{ and } m & & n\skelin m & \\
  t\skelin t' \mbox{ and } u\skelin u' & \Ra & (t\ u)\skelin(t'\ u') & \\
  t\skelin t' & \Ra & \lambda t\skelin\lambda t' & \\
  t\skelin t' & \Ra & t\skelin t'[s] & \mbox{if } \neg PR(s) \\
  t\skelin t' \mbox{ and } s\skelin s' & \Ra & t[s]\skelin t'[s'] \\
\end{array}
$$

$$
\begin{array}{cccc}
  & & \ua\skelin \ua & \\
  & & id\skelin id \\
  & & id\skelin s & \mbox{if } \neg PR(s) \\
  t\skelin t' \mbox{ and } s\skelin s' & \Ra & t\cdot s\skelin t'\cdot s' & \\
  s\skelin s' & \Ra & s \skelin s'\circ s_1 & \mbox{if } \neg PR(s_1) \\
  s\skelin s' & \Ra & s \skelin s_1\circ s' & \mbox{if } \neg PR(s_1) \\
  s_1\skelin s_1' \mbox{ and } s_2\skelin s_2' & \Ra & s_1\circ
  s_2\skelin s_1'\circ s_2' & \\
\end{array}
$$

\begin{example}
We have $t[t'\cdot id]\skelin t[1\cdot id][t'\cdot\ua]$.
\end{example}

We define the relation $\rel$ as follows.

\begin{definition}
  For all $t$ and $u$, $u\rel t\iff u\skelin t \mbox{ and }
  \sigma(t)=\sigma(u)$.
\end{definition}

Remark that we always have $t\rel t$.

Here follows several lemmas that will be used to prove the initialization
lemma. The first one says that the $\sigma$-normal form does not change when
one deletes a substitution $[id]$.

\begin{lemma}\label{l:init-id}
  For all $t$, we have $\sigma(t)=\sigma(t[id])$.
\end{lemma}

\begin{proof}
  See~\cite{AbaCarCurLev91}.
\end{proof}

\begin{lemma}\label{l:commute-shift}
  For all $t$, we have $\sigma(\Upshift{0}{1}{t}) =
  \Upshift{0}{1}{\sigma(t)}$.
\end{lemma}

\begin{proof}
  Since the application of $\Upshift{0}{1}{}$ changes only the values of
  free variable, its application is orthogonal to the reduction of substitutions
  that changes only the values of bound variables.
\end{proof}

The following lemma says that the $\sigma$-normal form of a term $t[\ua]$
is the same as that of $\Upshift{0}{1}{t}$.

\begin{lemma}\label{l:init-shift}
  For all $t$, we have $\sigma(\Upshift{0}{1}{t})=\sigma(t[\ua])$.
\end{lemma}

\begin{proof}
  We prove a more general result. Let $\Ua(s) = 1\cdot(s\circ\ua)$,
  we prove that for all $t$ and $i$, we have
  $\sigma(\Upshift{i}{1}{t})=\sigma(t[\Ua^i\!\!(\ua)])$. Since
  $\sigma(t[\Ua^i\!\!(\ua)]) = \sigma(\sigma(t)[\Ua^i\!\!(\ua)])$, it is
  enough to prove it for $t$ in $\sigma$-normal form. We proceed by
  induction on it.

  \begin{itemize}
  \item $t=u\ v$: then $\sigma((u\ v)[\Ua^i\!\!(\ua)])=
    \sigma((u[\Ua^i\!\!(\ua)])\ \sigma(v[\Ua^i\!\!(\ua)])$, and we conclude by
    induction hypothesis.
  \item $t=\lambda u$: then $\sigma((\lambda u)[\Ua^i\!\!(\ua)])=
    \lambda(\sigma(u[\Ua^{i+1}\!\!(\ua)]))$ and
    $\sigma(\Upshift{i}{1}{\lambda u}) =
    \lambda(\sigma(\Upshift{i+1}{1}{u}))$. We conclude by
    induction hypothesis.
  \item $t=1$: there are two cases,
    \begin{itemize}
    \item either $i=0$, then $\sigma(\Upshift{0}{1}{1}) =
      \sigma(2)$ and $\sigma(1[\ua]) = \sigma(2)$.
    \item or $i>0$, then $\sigma(\Upshift{i}{1}{1}) =
      \sigma(1) = 1$ and $\sigma(1[\Ua^i\!\!(\ua)]) =
      \sigma(1[1\cdot(\Ua^{i-1}\!\!(\ua)\circ\ua)]) =_{VarCons}
      \sigma(1) = 1$.
    \end{itemize}
  \item $t=n>1$: there are two cases,
    \begin{itemize}
    \item either $i<n$, then $\sigma(\Upshift{i}{1}{n}) =
      \sigma(n+1) = \sigma(\indice{n}) =_{Clos}
      1[\underset{n}{\underbrace{\ua\circ...\circ\ua}}]$ and

      $$
      \sigma(n[\Ua^i\!\!(\ua)]) =
      \sigma(\indice{n-1}[\Ua^i\!\!(\ua)])
      =_{Clos}
      \sigma(1[\underset{n-1}{\underbrace{\ua\circ...\circ\ua}}\circ\Ua^i\!\!(\ua)])
      $$

      We prove that this last term is equal to
      $1[\underset{n}{\underbrace{\ua\circ...\circ\ua}}]$ by
      induction on $i$:
      \begin{itemize}
      \item $i=0$:
        $\sigma(1[\underset{n-1}{\underbrace{\ua\circ...\circ\ua}}\circ\ua])
        = 1[\underset{n}{\underbrace{\ua\circ...\circ\ua}}]$
      \item $i>0$:
        $$
        \sigma(1[\underset{n-1}{\underbrace{\ua\circ...\circ\ua}}
        \circ\Ua^i\!\!(\ua)]) =
        \sigma(1[\underset{n-1}{\underbrace{\ua\circ...\circ\ua}}
        \circ(1\cdot(\Ua^{i-1}\!\!(\ua)\circ\ua))]) =
        \sigma(1[\underset{n-2}{\underbrace{\ua\circ...\circ\ua}}
        \circ\Ua^{i-1}\!\!(\ua)\circ\ua])
        $$

        By rule $Clos$, we have
        $\sigma(1[\underset{n-2}{\underbrace{\ua\circ...\circ\ua}}
        \circ\Ua^{i-1}\!\!(\ua)\circ\ua]) =
        \sigma(1[\underset{n-2}{\underbrace{\ua\circ...\circ\ua}}
        \circ\Ua^{i-1}\!\!(\ua)][\ua]) = \sigma(\sigma(1[\underset{n-2}{\underbrace{\ua\circ...\circ\ua}}
        \circ\Ua^{i-1}\!\!(\ua)])[\ua])$. Since $i<n$, then
        $i-1<n-1$ and we can apply the induction hypothesis on $i$,
        giving us $\sigma(\sigma(1[\underset{n-2}{\underbrace{\ua\circ...\circ\ua}}
        \circ\Ua^{i-1}\!\!(\ua)])[\ua]) =
        \sigma(1[\underset{n-1}{\underbrace{\ua\circ...\circ\ua}}][\ua])
        =_{Clos} 1[\underset{n}{\underbrace{\ua\circ...\circ\ua}}]$.
      \end{itemize}
    \item either $i\geq n$, then $\sigma(\Upshift{i}{1}{n}) =
      \sigma(n) = \sigma(\indice{n-1}) =_{Clos}
      1[\underset{n-1}{\underbrace{\ua\circ...\circ\ua}}]$ and

      $$
      \sigma(n[\Ua^i\!\!(\ua)]) =
      \sigma(\indice{n-1}[\Ua^i\!\!(\ua)])
      =_{Clos}
      \sigma(1[\underset{n-1}{\underbrace{\ua\circ...\circ\ua}}\circ\Ua^i\!\!(\ua)])
      $$

      We prove that this last term is equal to
      $1[\underset{n-1}{\underbrace{\ua\circ...\circ\ua}}]$ by
      induction on $i$:
      \begin{itemize}
      \item $i=0$: impossible since $i\geq n>1$.
      \item $i=1$: impossible since $i\geq n>1$.
      \item $i=2$: we must have $n=2=1[\ua]$, giving us
        $$
        \begin{array}{rll}
          \sigma(1[\ua\circ\Ua^2\!\!(\ua)]) & = &
          \sigma(1[\ua\circ(1\cdot((1\cdot(\ua\circ\ua))\circ\ua))]) \\
          & =_{ShiftCons} & \sigma(1[(1\cdot(\ua\circ\ua))\circ\ua]) \\
          & =_{Map} & \sigma(1[1[\ua]\cdot(\ua\circ\ua\circ\ua)]) \\
          & =_{VarCons} & 1[\ua]
        \end{array}
        $$
      \item $i>2$:
        $$
        \sigma(1[\underset{n-1}{\underbrace{\ua\circ...\circ\ua}}
        \circ\Ua^i\!\!(\ua)]) =
        \sigma(1[\underset{n-1}{\underbrace{\ua\circ...\circ\ua}}
        \circ(1\cdot(\Ua^{i-1}\!\!(\ua)\circ\ua))]) =
        \sigma(1[\underset{n-2}{\underbrace{\ua\circ...\circ\ua}}
        \circ\Ua^{i-1}\!\!(\ua)\circ\ua])
        $$

        By rule $Clos$, we have
        $\sigma(1[\underset{n-2}{\underbrace{\ua\circ...\circ\ua}}
        \circ\Ua^{i-1}\!\!(\ua)\circ\ua]) =
        \sigma(1[\underset{n-2}{\underbrace{\ua\circ...\circ\ua}}
        \circ\Ua^{i-1}\!\!(\ua)][\ua])$. Since $i\geq n$, then
        $i-1\geq n-1$ and we can apply the induction hypothesis on
        $i$, giving us $\sigma(1[\underset{n-2}{\underbrace{\ua\circ...\circ\ua}}
        \circ\Ua^{i-1}\!\!(\ua)][\ua]) =
        \sigma(1[\underset{n-2}{\underbrace{\ua\circ...\circ\ua}}][\ua])
        =_{Clos} 1[\underset{n-1}{\underbrace{\ua\circ...\circ\ua}}]$.
      \end{itemize}
    \end{itemize}
  \end{itemize}
\end{proof}

We also need a lemma to equalize $\sigma$-normal forms.

\begin{lemma}\label{l:egal-sigma}
  For all $t$, $u$, $s$ and $s'$ in $\sigma$-normal form, if
  $\sigma(t[1\cdot(s\circ\ua)])=\sigma(t'[1\cdot(s'\circ\ua)])$ then
  $\sigma(t[u\cdot s])=\sigma(t'[u\cdot s'])$.
\end{lemma}

\begin{proof}
  Easy induction on $t$.
\end{proof}

We can now prove our initialization lemma.

\begin{lemma}[Initialisation]\label{l:Ateb-sim-ls}
  For all $t$, there exists $u$ such that $\mathit{Ateb}(t) \ra_B^* u$ and $u\rel
  t$.
\end{lemma}

\begin{proof}
  By induction on the number of reduction steps of $\mathit{Ateb}(t)$ and by case analysis of $t$.
  \begin{itemize}
    \item If $t=n$, then $\mathit{Ateb}(t)=n$ and we conclude with $u=n$.
    \item If $t=(t_1\ t_2)$, then $\mathit{Ateb}(t)=(\mathit{Ateb}(t_1)\
      \mathit{Ateb}(t_2))$. By induction hypothesis, there exists $u_1$ and
      $u_2$ such that $\mathit{Ateb}(t_1) \ra_B^* u_1$ and $u_1\rel t_1$ and
      $\mathit{Ateb}(t_2) \ra_B^* u_2$ and $u_2\rel t_2$. We conclude with
      $u=(u_1\ u_2)$.
    \item If $t=\lambda t'$, then $\mathit{Ateb}(t)=\lambda\mathit{Ateb}(t')$. By induction
      hypothesis, there exists $u'$ such that $\mathit{Ateb}(t')
      \ra_B^* u'$. We conclude with $u=\lambda u'$.
    \item If $t=t'[id]$, then $\mathit{Ateb}(t)=\mathit{Ateb}(t')$. By
      induction hypothesis, there exists $u'$ such that $\mathit{Ateb}(t')
      \ra_B^* u'$. We take $u=u'$ and we conclude with the help of Lemma~\ref{l:init-id}.
    \item If $t=t'[\ua]$, then
      $\mathit{Ateb}(t)=\Upshift{0}{1}{\mathit{Ateb}(t')}$. By induction
      hypothesis, there exists $u'$ such that $\mathit{Ateb}(t')
      \ra_B^* u'$. We take $u=\Upshift{0}{1}{u'}$ and we conclude with the help
      of Lemmas~\ref{l:commute-shift} and~\ref{l:init-shift}.
    \item If $t=t'[s\circ s']$, then $\mathit{Ateb}(t)=\mathit{Ateb}(t'[s][s'])$. By
      hypothesis, there exists $u'$ such $\mathit{Ateb}(t'[s][s']) \ra_B^* u'$.
      Four cases arise with respect to the values of $PR(s)$ and $PR(s')$,
      in all those cases, we can conclude with $u=u'$.
    \item If $t=t_1[t_2\cdot s]$, then $\mathit{Ateb}(t)=\mathit{Ateb}((\lambda
      t_1)[s])\ \mathit{Ateb}(t_2)$. By induction hypothesis, there exists
      $u_1$ and $u_2$ such that $\mathit{Ateb}((\lambda t_1)[s]) \ra_B^* u_1$ and
      $u_1\rel (\lambda t_1)[s]$ and $\mathit{Ateb}(t_2) \ra_B^* u_2$ and
      $u_2\rel t_2$. There are two cases with respect to the form of $u_1$.
      \begin{itemize}
      \item If $u_1=\lambda v_1$ (and so $\neg
        PR(s)$), then we take $u=v_1[u_2\cdot id]$. We must check that $u\rel t$
        and the difficulty resides in the proof of
        $\sigma(u)=\sigma(t)$. By hypothesis, we have $\sigma(\lambda
        v_1) = \sigma((\lambda t_1)[s])$. It is obvious that $\sigma(\lambda
        v_1) = \sigma((\lambda v_1)[id]) =
        \lambda(\sigma(\sigma(v_1)[1\cdot(id\circ\ua)]))$. On the other hand, we have
        $\sigma((\lambda t_1)[s]) =
        \lambda(\sigma(\sigma(t_1)[1\cdot(\sigma(s)\circ\ua)]))$, and it gives us
        $\sigma(\sigma(v_1)[1\cdot(id\circ\ua)]) =
        \sigma(\sigma(t_1)[1\cdot(\sigma(s)\circ\ua)])$. We have the required conditions
        to apply Lemma ~\ref{l:egal-sigma} with the term $\sigma(u_2)$
        that is equal to $\sigma(t_2)$ by hypothesis, and we get
        $\sigma(\sigma(v_1)[\sigma(u_2)\cdot id]) =
        \sigma(\sigma(t_1)[\sigma(t_2)\cdot s])$ which concludes this point.
      \item If $u_1=(\lambda v_1)[s_1]$, then we take
        $u=v_1[u_2\cdot s_1]$ and we conclude similarly to the previous point
        with the help of Lemma~\ref{l:egal-sigma}.
      \end{itemize}
  \end{itemize}
\end{proof}

\subsubsection{Simulation}

We can now perform the simulation.

\begin{lemma}[Simulation]\label{l:simul-ls}
  For all $t$ reducing by rule $B$ to $t'$,
  for all $u\rel t$, there exists $u'$ such that
  $u$ reduces in one step to  $u'$ and $u'\rel t'$.
  For all he other rules, for all $t$ reducing to $t'$,
  for all $u\rel t$, there exists $u'$ such that
  $u$ reduces in zero or some steps to $u'$
  and $u'\rel t'$.
\end{lemma}

\begin{proof}
  For all the rules apart from $B$ ({\em i.e.} $\sigma$), the proof is simple.
  $u\rel t$ gives us $u\skelin t$ and $\sigma(u) =
  \sigma(t)$, and, on the other hand, $t\ra_{\sigma} t'$ implies
  $\sigma(t)=\sigma(t')$. Two cases are possible with respect to the fact
  that the redex appears also in $u$. If not, we take $u'=u$
  and we directly conclude. Else, we reduce it with the same rule and we conclude
  with $\sigma(u') = \sigma(u) = \sigma(t')$.

  \medskip

  It's more complicated for the rule $B$. The hypothesis is the same but we are sure that
  the redex appears in $u$, that was the point of defining the relation $\skelin$
  with the help of the predicate $PR$. We then have $u\ra_B u'$ and we want to
  prove $u'\rel t'$. Even if it is obvious that $u'\skelin t'$ comes directly
  from $u\skelin t$, it is not the case for the equality of the $\sigma$-normal
  forms. We want $\sigma(u')=\sigma(t')$ with the hypothesis $\sigma(u)=\sigma(t)$.
  We take $t=C[(\lambda v)\ w]$, which gives us $t'=C[v[w\cdot id]]$ and $u=C'[(\lambda v')\
  w']$. Two cases are possible:
  \begin{itemize}
  \item the redex $(\lambda v)\ w$ does not appear in $\sigma(t)$.
    It means that the calculus of $\sigma(t)$ can be split as follows:

    $$
    C[(\lambda v)\ w]\ra^*_{\sigma} C_1[\ua\circ(C_2[(\lambda v)\
    w]\cdot s)] \ra_{ShiftCons} C_1[s] \ra^*_{\sigma} \sigma(t)
    $$

    Since $\sigma(t)=\sigma(u)$, the same occurs for $u$. Similarly for the redex,
    the reduct will be erased from $t'$ and from $u'$ and we get $\sigma(u')=\sigma(t')$.

  \item the redex $(\lambda v)\ w$ does appear in
    $\sigma(t)$. We will write, for all $t$, $\ul{t}$ for $\sigma(t)$,
    in order to clarify the presentation of the calculi. We have
    the following equalities:

    $$
    \begin{array}{rcl}
      \ul{t} & = & \sigma(C[(\lambda v)\ w]) \\
      & = & C_1[\sigma(((\lambda v)\ w)[\ul{s}])] \\
      & = & C_1[(\lambda\sigma(\ul{v}[1\cdot(\ul{s}\circ\ua)]))\ \sigma(\ul{w}[\ul{s}])] \\
    \end{array}
    $$

    And, similarly,
    $\ul{u}=C'_1[(\lambda\sigma(\ul{v'}[1\cdot(\ul{s'}\circ\ua)]))\
    \sigma(\ul{w'}[\ul{s'}])]$. From $\ul{t}=\ul{u}$ we deduce
    $C_1=C'_1$, $\sigma(\ul{v}[1\cdot(\ul{s}\circ\ua)]) =
    \sigma(\ul{v'}[1\cdot(\ul{s'}\circ\ua)])$ and
    $\sigma(\ul{w}[\ul{s}]) = \sigma(\ul{w'}[\ul{s'}])$. We now look
    at $\ul{t'}$ and $\ul{u'}$:

    $$
    \begin{array}{rl}
      \ul{t'} = & \sigma(C[v[w\cdot id]]) \\
      = & C_1[\sigma(\ul{v}[\ul{w}\cdot id][\ul{s}])] \\
      =_{Clos} & C_1[\sigma(\ul{v}[(\ul{w}\cdot id)\circ\ul{s}])] \\
      =_{Map} & C_1[\sigma(\ul{v}[\ul{w}[\ul{s}]\cdot \ul{s}])] \\
      = & C_1[\sigma(\ul{v}[\sigma(\ul{w}[\ul{s}])\cdot \ul{s}])] \\
    \end{array}
    $$

    And similarly,
    $\ul{u'}=C'_1[\sigma(\ul{v'}[\sigma(\ul{w'}[\ul{s'}])\cdot
    \ul{s'}])]$. From the preceding equalities, we deduce
    $\ul{u'}=C'_1[\sigma(\ul{v'}[\sigma(\ul{w}[\ul{s}])\cdot
    \ul{s'}])]$, and we can conclude with the help of Lemma~\ref{l:egal-sigma}.
  \end{itemize}
\end{proof}

\begin{lemma}
  For all terms $t$, if $u\rel t$ and $u\in \Lambda^X_{SN}$, then $t\in \Lambda^X_{SN}$.
\end{lemma}

\begin{proof}
  By the simulation Lemma~\ref{l:simul-ls}, and since the $\sigma$-calculus is
  terminating~\cite{AbaCarCurLev91}, if we have an infinite derivation
  of $t$, then we can also build one in $u$, and that gives us a contradiction.
\end{proof}

Since the $\mathit{Ateb}(t)$ function returns a term $t'$ that reduces to $u\rel t$
(by Lemma~\ref{l:Ateb-sim-ls}), we know this technique can be applied to this calculus.

\section{$\lambda_{\sigma n}$-calculus}

In this section, we study a version with names of the $\lsigma$-calculus~\cite{AbaCarCurLev91}.
The same remarks will apply here as regards to the application of this technique.

\subsection{Definition}

Terms of the $\lambda_{\sigma n}$-calculus are given by the following grammar:

$$
\begin{array}{l}
  t ::= x\ |\ (t\ t)\ |\ \lambda x.t\ |\ t[s] \\
  s ::= id\ |\ (t/x)\cdot s\ |\ s\circ s \\
\end{array}
$$

Here follows the reduction rules:

$$
\begin{array}{llll}
  (\lambda x.t) u & \ra_{B} & t[(u/x)\cdot id] & \\ \\
  (t\ u)[s] & \ra_{App} & (t[s])\ (u[s]) & \\
  (\lambda x.t)[s] & \ra_{Lambda} & \lambda y.(t[(y/x)\cdot
  s]) & \mbox{with $y$ fresh} \\
  x[id] & \ra_{VarId} & x & \\
  x[(t/x)\cdot s] & \ra_{VarCons1} & t & \\
  x[(t/y)\cdot s] & \ra_{VarCons2} & x[s] & (x\not=y)\\
  t[s][s'] & \ra_{Clos} & t[s\circ s'] \\ \\
  id\circ s & \ra_{IdL} & s \\
  ((t/x)\cdot s)\circ s' & \ra_{Map} & (t[s']/x)\cdot(s\circ s') \\
  (s_1\circ s_2)\circ s_3 & \ra_{Ass} & s_1\circ(s_2\circ s_3) \\
\end{array}
$$

\subsection{Towards strong normalization}

We define the $\mathit{Ateb}$ function as follows:

$$
\begin{array}{lll}
  \mathit{Ateb}(x) & = & x \\
  \mathit{Ateb}(t\ u) & = & Ateb(t)\ Ateb(u) \\
  \mathit{Ateb}(\lambda x.t) & = & \lambda x.Ateb(t) \\
  \mathit{Ateb}(t[id]) & = & Ateb(t) \\
  \mathit{Ateb}(t[s\circ s']) & = & Ateb(t[s][s']) \\
  \mathit{Ateb}(t[(t'/x)\cdot s]) & = & Ateb((\lambda x.t)[s])\ Ateb(t') \\
\end{array}
$$

It is obvious that for any $t$,
$\mathit{Ateb}(t)$ does not contain substitutions. We must check that the
term we obtain is typeable.

\begin{lemma}
  $$\Gamma\vdash t:A \;\;\Ra\;\; \Gamma\vdash \mathit{Ateb}(t):A$$
\end{lemma}

\begin{proof}
  By induction on $t$.

  \begin{itemize}
  \item $t=x$ and

    $$
    \infer{x:A,\Delta\vdash x:A}{}
    $$

    We have $\mathit{Ateb}(t)=x$ and the same typing derivation.

  \item $t=(u\ v)$ and

    $$
    \infer
      {\Gamma\vdash (u\ v):A}
      {\Gamma\vdash u:B\ra A & \Gamma\vdash v:B}
    $$

    By induction hypothesis, we have $\Gamma\vdash \mathit{Ateb}(u):B\ra A$ and
    $\Gamma\vdash \mathit{Ateb}(v):B$. We can type $\mathit{Ateb}(t)=\mathit{Ateb}(u)\
    \mathit{Ateb}(v)$ as follows

    $$
    \infer
      {\Gamma\vdash (\mathit{Ateb}(u)\ \mathit{Ateb}(v)):A}
      {\Gamma\vdash \mathit{Ateb}(u):B\ra A & \Gamma\vdash \mathit{Ateb}(v):B}
    $$

  \item $t=\lambda x.u$ and

    $$
    \infer
      {\Gamma\vdash \lambda x.u:B\ra A}
      {x:B,\Gamma\vdash u:A}
    $$

    By induction hypothesis, we have $\Gamma,x:B\vdash \mathit{Ateb}(u):A$. We can type
    $\mathit{Ateb}(t)=\lambda x.\mathit{Ateb}(u)$ as follows

    $$
    \infer
      {\Gamma\vdash \lambda x.\mathit{Ateb}(u):B\ra A}
      {x:B,\Gamma\vdash \mathit{Ateb}(u):A}
    $$

  \item $t=u[id]$ and

    $$
    \infer
      {\Gamma\vdash u[id]:A}
      {\Gamma\vdash id\trg\Gamma & \Gamma\vdash u:A}
    $$

    We directly conclude by induction hypothesis.

  \item $t=u[(v/x)\cdot s]$ and

    $$
    \infer
      {\Gamma\vdash u[(v/x)\cdot s]:A}
      {\infer
        {\Gamma\vdash (v/x)\cdot s\trg x:B,\Gamma'}
        {\Gamma\vdash s\trg\Gamma' & \Gamma\vdash v:B} &
       x:B,\Gamma'\vdash u:A}
    $$

    By induction hypothesis, we have $\Gamma\vdash\mathit{Ateb}((\lambda x.u)[s]):B\ra
    A$ and $\Gamma\vdash\mathit{Ateb}(v):B$. We conclude with the following typing
    derivation

    $$
    \infer
      {\Gamma\vdash\mathit{Ateb}((\lambda x.u)[s])\ \mathit{Ateb}(v):A}
      {\Gamma\vdash\mathit{Ateb}((\lambda x.u)[s]):B\ra A &
       \Gamma\vdash\mathit{Ateb}(v):B}
    $$

  \item $t=u[s\circ s']$, we conclude directly by induction hypothesis.
  \end{itemize}
\end{proof}

\subsection{Definition of the relation $\rel$}

We proceed as in the previous section, but more easily since there is no $\ua$.
We use the same notion of redexability (see~\ref{def:redexabilite}) and
the relation $\skelin$ is define similarly (without the $\ua$). We define the relation
$\rel$ as follows.

\begin{definition}
  For all $t$ and $u$, $u\rel t\iff u\skelin t \mbox{ and }
  \sigma(t)=\sigma(u)$.
\end{definition}

We will use Lemma~\ref{l:init-id} and we need a new formulation of the Lemma~\ref{l:egal-sigma}.

\begin{lemma}\label{l:egal-sigma-lsn}
  For all $t$, $u$, $s$ and $s'$ in $\sigma$-normal form, if
  $\sigma(t[(y/x)\cdot s])=\sigma(t'[(y/x)\cdot s'])$ then
  $\sigma(t[(u/x)\cdot s])=\sigma(t'[(u/x)\cdot s'])$.
\end{lemma}

\begin{proof}
  Easy induction.
\end{proof}

Here follows our initialization Lemma.

\begin{lemma}[Initialisation]\label{l:Ateb-sim-lsn}
  For all $t$, there exists $u$ such that $\mathit{Ateb}(t) \ra_B^* u$ and $u\rel
  t$.
\end{lemma}

\begin{proof}
  By induction on $t$.
  \begin{itemize}
    \item If $t=x$, then $\mathit{Ateb}(t)=x$ and we conclude with $u=x$.
    \item If $t=(t_1\ t_2)$, then $\mathit{Ateb}(t)=(\mathit{Ateb}(t_1)\
      \mathit{Ateb}(t_2))$. By induction hypothesis, there exists $u_1$ and
      $u_2$ such that $\mathit{Ateb}(t_1) \ra_B^* u_1$ and $u_1\rel t_1$ and
      $\mathit{Ateb}(t_2) \ra_B^* u_2$ and $u_2\rel t_2$. We conclude with
      $u=(u_1\ u_2)$.
    \item If $t=\lambda x.t'$, then $\mathit{Ateb}(t)=\lambda x.\mathit{Ateb}(t')$.
      By induction hypothesis, there exists $u'$ such that $\mathit{Ateb}(t')
      \ra_B^* u'$. We conclude with $u=\lambda x.u'$.
    \item If $t=t'[id]$, then $\mathit{Ateb}(t)=\mathit{Ateb}(t')$.
      By induction hypothesis, there exists $u'$ such that $\mathit{Ateb}(t')
      \ra_B^* u'$. We take $u=u'$ and we conclude with the help of Lemma~\ref{l:init-id}.
    \item If $t=t'[s\circ s']$, then $\mathit{Ateb}(t)=\mathit{Ateb}(t'[s][s'])$.
      By induction hypothesis, there exists $u'$ such that $\mathit{Ateb}(t'[s][s']) \ra_B^* u'$.
      We conclude with $u=u'$.
    \item If $t=t_1[(t_2/x)\cdot s]$, then $\mathit{Ateb}(t)=\mathit{Ateb}((\lambda
      x.t_1)[s])\ \mathit{Ateb}(t_2)$. By induction hypothesis, there exists
      $u_1$ and $u_2$ such that $\mathit{Ateb}((\lambda x.t_1)[s]) \ra_B^* u_1$ and
      $u_1\rel (\lambda x.t_1)[s]$ and $\mathit{Ateb}(t_2) \ra_B^* u_2$ and
      $u_2\rel t_2$. There are two cases with respect to the form of $u_1$.
      \begin{itemize}
      \item If $u_1=\lambda x.v_1$ (and so $\neg
        PR(s)$), then we take $u=v_1[(u_2/x)\cdot id]$. We need to check that $u\rel t$
        and the difficulty resides in the proof of
        $\sigma(u)=\sigma(t)$. By hypothesis, we have $\sigma(\lambda
        x.v_1) = \sigma((\lambda x.t_1)[s])$. It is obvious that $\sigma(\lambda
        x.v_1) = \sigma((\lambda x.v_1)[id]) =
        \lambda y.(\sigma(\sigma(v_1)[(y/x)\cdot id]))$. On the other hand, we have
        $\sigma((\lambda x.t_1)[s]) =
        \lambda y.(\sigma(\sigma(t_1)[(y/x)\cdot \sigma(s)]))$, which gives us
        $\sigma(\sigma(v_1)[(y/x)\cdot id]) =
        \sigma(\sigma(t_1)[(y/x)\cdot \sigma(s)])$. We can apply Lemma~\ref{l:egal-sigma-lsn}
        with $\sigma(u_2)$ that is equal to $\sigma(t_2)$ by hypothesis, that gives us
        $\sigma(\sigma(v_1)[(\sigma(u_2)/x)\cdot id]) =
        \sigma(\sigma(t_1)[(\sigma(t_2)/x)\cdot s])$ and we can conclude.
      \item If $u_1=(\lambda x.v_1)[s_1]$, then we take
        $u=v_1[(u_2/x)\cdot s_1]$ and we conclude similarly with
        the help of Lemma~\ref{l:egal-sigma-lsn}.
      \end{itemize}
  \end{itemize}
\end{proof}

\subsubsection{Simulation}

The simulation will be similar to that defined for $\lsigma$.

\begin{lemma}[Simulation]\label{l:simul-lsn}
  For all $t$ reducing with the rule $B$ to $t'$,
  for all $u\rel t$, there exists $u'$ such that
  $u$ reduces in one step to $u'$ and $u'\rel t'$.
  For all the other rules, for all $t$ reducing to $t'$,
  for all $u\rel t$, there exists $u'$ such that
  $u$ reduces in zero or some steps to $u'$ and $u'\rel t'$.
\end{lemma}

\begin{proof}
  For all the rules except $B$ (namely $\sigma$), the proof is simple.
  $u\rel t$ gives us $u\skelin t$ and $\sigma(u) =
  \sigma(t)$, and, on the other hand, $t\ra_{\sigma} t'$ implies
  $\sigma(t)=\sigma(t')$. There are two cases with respect to the fact
  that the redex appears in $u$. If not, we take $u'=u$
  and we conclude directly. Else, we reduce it with the same rule and
  we conclude with $\sigma(u') = \sigma(u) = \sigma(t')$.

  \medskip

  For the rule $B$, it's more complicated. The hypothesis is he same, but
  we are sure that the redex appears in $u$, that was the point of defining the relation $\skelin$
  with the help of the predicate $PR$. We then have $u\ra_B u'$ and we want to
  prove $u'\rel t'$. Even if it is obvious that $u'\skelin t'$ comes directly
  from $u\skelin t$, it is not the case for the equality of the $\sigma$-normal
  forms. We want $\sigma(u')=\sigma(t')$ with the hypothesis $\sigma(u)=\sigma(t)$.
  We take $t=C[(\lambda x.v)\ w]$, which
  gives us $t'=C[v[(w/x)\cdot id]]$ and $u=C'[(\lambda x.v')\
  w']$. Two cases are possible:
  \begin{itemize}
  \item the redex $(\lambda x.v)\ w$ does not appear in
    $\sigma(t)$. That means that the calculus of $\sigma(t)$ can be split as
    follows:

    $$
    C[(\lambda x.v)\ w]\ra^*_{\sigma} C_1[y[((C_2[(\lambda x.v)\
    w])/x)\cdot s]] \ra_{VarCons2} C_1[y[s]] \ra^*_{\sigma} \sigma(t)
    $$

    Since $\sigma(t)=\sigma(u)$, it occurs similarly for
    $u$. As for the redex, the reduct will be erased from
    $t'$ and from $u'$ and we get $\sigma(u')=\sigma(t')$.

  \item the redex $(\lambda x.v)\ w$ does appear in
    $\sigma(t)$. We will write, for all $t$, $\ul{t}$ for $\sigma(t)$,
    in order to clarify the presentation of the calculi. We have
    the following equalities:

    $$
    \begin{array}{rcl}
      \ul{t} & = & \sigma(C[(\lambda x.v)\ w]) \\
      & = & C_1[\sigma(((\lambda x.v)\ w)[\ul{s}])] \\
      & = & C_1[(\lambda y.\sigma(\ul{v}[(y/x)\cdot\ul{s}]))\ \sigma(\ul{w}[\ul{s}])] \\
    \end{array}
    $$

    And, similarly,
    $\ul{u}=C'_1[(\lambda y.\sigma(\ul{v'}[(y/x)\cdot\ul{s'}]))\
    \sigma(\ul{w'}[\ul{s'}])]$. From $\ul{t}=\ul{u}$ we deduce
    $C_1=C'_1$, $\sigma(\ul{v}[(y/x)\cdot\ul{s}]) =
    \sigma(\ul{v'}[(y/x)\cdot\ul{s'}])$ and
    $\sigma(\ul{w}[\ul{s}]) = \sigma(\ul{w'}[\ul{s'}])$. We now look
    at $\ul{t'}$ and $\ul{u'}$:

    $$
    \begin{array}{rl}
      \ul{t'} = & \sigma(C[v[w\cdot id]]) \\
      = & C_1[\sigma(\ul{v}[\ul{w}\cdot id][\ul{s}])] \\
      =_{Clos} & C_1[\sigma(\ul{v}[(\ul{w}\cdot id)\circ\ul{s}])] \\
      =_{Map} & C_1[\sigma(\ul{v}[\ul{w}[\ul{s}]\cdot \ul{s}])] \\
      = & C_1[\sigma(\ul{v}[\sigma(\ul{w}[\ul{s}])\cdot \ul{s}])] \\
    \end{array}
    $$

    And similarly,
    $\ul{u'}=C'_1[\sigma(\ul{v'}[\sigma(\ul{w'}[\ul{s'}])\cdot
    \ul{s'}])]$. From the preceding equalities, we deduce
    $\ul{u'}=C'_1[\sigma(\ul{v'}[\sigma(\ul{w}[\ul{s}])\cdot
    \ul{s'}])]$, and we can conclude with the help of Lemma~\ref{l:egal-sigma-lsn}.
  \end{itemize}
\end{proof}

\begin{lemma}
  For all terms $t$, if $u\rel t$ and $u\in \Lambda^X_{SN}$, then $t\in \Lambda^X_{SN}$.
\end{lemma}

\begin{proof}
  By the simulation Lemma~\ref{l:simul-lsn}, and since the $\sigma$-calculus is
  terminating~\cite{AbaCarCurLev91}, if we have an infinite derivation
  of $t$, then we can also build one in $u$, and that gives us a contradiction.
\end{proof}

Since the $\mathit{Ateb}(t)$ function returns a term $t'$ that reduces to $u\rel t$
(by Lemma~\ref{l:Ateb-sim-lsn}), we know this technique can be applied to this calculus.

\section{$\lbdch$-calculus}

\newcommand{\coupe}[2]{\langle #1 | #2 \rangle}
\newcommand{\Varb}{Var^\perp}
\newcommand{\Dom}[1]{{\ensuremath{\mathit{Dom}}(#1)}}

The $\lbdch$-calculus is a symmetric non-deterministic calculus
that comes from classical logic. Its terms represent proof in classical sequent calculus.
We can add to it explicit substitutions ``à la'' $\lambda{\tt x}$.

\subsection{Definition}

We have four syntactic categories: terms, contexts, commands and substitutions ;
respectively denoted $v$, $e$, $c$ and $\tau$.
We give to variable sets:
$Var$ is the set of term variables (denoted $x$, $y$, $z$ etc.);
$\Varb$ is the set of context variables (denoted $\alpha$, $\beta$,
$\gamma$ etc.). We will denote by $*$ a variable for which the set to which it belong
does not care, and by $t$ an undetermined syntactic object among $v$, $e$ and $c$.

The syntax of the $\lbdch$-calculus is given by the following grammar:

$$
\begin{array}{l}
  c\ ::=\ \coupe{v}{e}\ |\ c\tau \\
  v\ ::=\ x\ |\ \lambda x.v\ |\ e\cdot v\ |\ \mu\alpha.c\ |\ v\tau \\
  e\ ::=\ \alpha\ |\ \alpha\lambda.e\ |\ v\cdot e\ |\ \tilde\mu x.c\
  |\ e\tau \\
  \tau ::= [x\la v]\ |\ [\alpha\la e] \\
\end{array}
$$

The source $\Dom{\tau}$ of $\tau$ is $x$ if
$\tau=[x \la v]$ and $\alpha$ if
$\tau=[\alpha\la e]$. The substituend $S(\tau)$ is $v$ and $e$ respectively.

The reduction rules are given below. Remark that the rules $(\mu)$ and $(\tilde\mu)$
gives a critical pair:

$$
\begin{array}{lrcll}
  (\beta) & \coupe{\lambda x.v}{v'\cdot e} & \ra &
  \coupe{v'}{\tilde\mu x.\coupe{v}{e}} \\
  (\tilde\beta) & \coupe{e'\cdot v}{\alpha\lambda.e} & \ra &
  \coupe{\mu\alpha.\coupe{v}{e}}{e'} \\
  (\mu) & \coupe{\mu\alpha.c}{e} & \ra & c[\alpha\la e] & \\
  (\tilde\mu) & \coupe{v}{\tilde\mu x.c} & \ra &
  c[x\la v] & \\ \\
  (c\tau) & \coupe{v}{e}\tau & \ra & \coupe{v\tau}{e\tau} & \\
  (x\tau 1) & x[x \la v] & \ra & v & \\ 
  (x\tau 2) & x\tau & \ra & x & \mbox{If } x\not\in \Dom{\tau} \\ 
  (\alpha\tau 1) & \alpha[\alpha\la e] & \ra & e & \\ 
  (\alpha\tau 2) & \alpha \tau & \ra & \alpha  & \mbox{If } \alpha \not\in \Dom{\tau} \\ 
  (\cdot\tau) & (v\cdot e)\tau & \ra & (v\tau)\cdot(e\tau) & \\
  (\tilde{\cdot}\tau) & (e\cdot v)\tau & \ra & (e\tau)\cdot(v\tau) & \\
  (\lambda\tau) & (\lambda x.v)\tau & \ra & \lambda
  x.(v\tau) & \\
  (\tilde\lambda\tau) & (\alpha\lambda.e)\tau & \ra &
  \alpha\lambda.(e\tau) & \\
  (\mu\tau) & (\mu\alpha.c)\tau & \ra & \mu\alpha.(c\tau) & \\
  (\tilde\mu\tau) & (\tilde\mu x.c)\tau & \ra & \tilde\mu x.(c\tau) & \\
\end{array}
$$

For the rules $(\mu\tau)$ and $(\tilde\lambda\tau)$
(resp. $(\tilde\mu\tau)$ and $(\lambda\tau)$) we might perform
$\alpha$-conversion on the bound variable $\alpha$ (resp. $x$) if necessary.
We add two simplification rules:

$$
\begin{array}{lrcll}
  (sv) & \mu\alpha.\coupe{v}{\alpha} & \ra & v & \mbox{Si } \alpha \not\in v \\ 
  (se) & \tilde\mu x.\coupe{x}{e} & \ra & e & \mbox{Si } x \not\in e 
\end{array}
$$

\medskip

Here follows the typing rules:

$$
\infer{\coupe{v}{e}:(\Gamma\vdash\Delta)}{\Gamma\vdash v:A|\Delta &
  \Gamma|e:A\vdash \Delta}
$$

$$
\begin{array}{rl}
  \infer{\Gamma|\alpha:A\vdash\Delta,\alpha:A}{} &
  \infer{\Gamma,x:A\vdash\Delta|x:A}{} \\ \\
  \infer{\Gamma|\alpha\lambda.e:A-B\vdash\Delta}{\Gamma|e:B\vdash\alpha:A,\Delta} &
  \infer{\Gamma\vdash\lambda x.v:A\ra B|\Delta}{\Gamma,x:A\vdash
    v:B|\Delta} \\ \\
  \infer{\Gamma|v\cdot e:A\ra B\vdash\Delta}{\Gamma\vdash
    v:A|\Delta & \Gamma|e:B\vdash\Delta} &
  \infer{\Gamma\vdash e\cdot v:A-B|\Delta}{\Gamma\vdash
    v:B|\Delta & \Gamma|e:A\vdash\Delta} \\ \\
  \infer{\Gamma|\tilde\mu x.c :A\vdash
    \Delta}{c:(\Gamma,x:A\vdash\Delta)} &
  \infer{\Gamma\vdash \mu\alpha.c
    :A|\Delta}{c:(\Gamma\vdash\alpha:A,\Delta)}
\end{array}
$$

$$
\begin{array}{rl}
  \infer{[x\la v]:(\Gamma,x:A\vdash\Delta)\Ra(\Gamma\vdash\Delta)}
  {\Gamma\vdash v:A|\Delta} &
  \infer{[\alpha\la e]:(\Gamma\vdash\alpha:A,\Delta)\Ra(\Gamma\vdash\Delta)}
  {\Gamma|e:A\vdash\Delta} \\ \\
  \infer{\Gamma'|e\tau:A\vdash\Delta'}{\Gamma|e:A\vdash\Delta & \tau:(\Gamma\vdash\Delta)\Ra(\Gamma'\vdash\Delta')} &
  \infer{\Gamma'\vdash v\tau:A|\Delta'}{\Gamma\vdash v:A|\Delta & \tau:(\Gamma\vdash\Delta)\Ra(\Gamma'\vdash\Delta')} \\
\end{array}
$$

$$
\infer{c\tau:(\Gamma'\vdash \Delta')}{c:(\Gamma\vdash\Delta) &
  \tau:(\Gamma\vdash\Delta)\Ra(\Gamma'\vdash\Delta')}
$$

\subsection{Strong normalization}

We define the $\mathit{Ateb}$ function as follows:

$$
\begin{array}{llll}
  \mathit{Ateb}(x) & = & x & \\
  \mathit{Ateb}(\alpha) & = & \alpha & \\
  \mathit{Ateb}(\coupe{v}{e}) & = & \coupe{\mathit{Ateb}(v)}{\mathit{Ateb}(e)} & \\
  \mathit{Ateb}(\lambda x.v) & = & \lambda x.\mathit{Ateb}(v) & \\
  \mathit{Ateb}(\alpha\lambda.e) & = & \alpha\lambda.\mathit{Ateb}(e) & \\
  \mathit{Ateb}(\mu\alpha.c) & = & \mu\alpha.\mathit{Ateb}(c) & \\
  \mathit{Ateb}(\tilde\mu x.c) & = & \tilde\mu x.\mathit{Ateb}(c) & \\
  \mathit{Ateb}(e\cdot v) & = & \mathit{Ateb}(e)\cdot\mathit{Ateb}(v) & \\
  \mathit{Ateb}(v\cdot e) & = & \mathit{Ateb}(v)\cdot\mathit{Ateb}(e) & \\ \\

  \mathit{Ateb}(c[x\la v]) & = & \coupe{\mathit{Ateb}(v)}{\tilde\mu x.\mathit{Ateb}(c)} & \\
  \mathit{Ateb}(c[\alpha\la e]) & = & \coupe{\mu\alpha.\mathit{Ateb}(c)}{\mathit{Ateb}(e)} & \\
  \mathit{Ateb}(v[x\la v']) & = & \mu\alpha.\coupe{\lambda x.\mathit{Ateb}(v)}{\mathit{Ateb}(v')\cdot\alpha} &
  \mbox{with } \alpha \mbox{ fresh} \\
  \mathit{Ateb}(v[\alpha\la e]) & = & \mu\beta.\coupe{\mu\alpha.\coupe{\mathit{Ateb}(v)}{\beta}}{\mathit{Ateb}(e)} &
  \mbox{with } \beta \mbox{ fresh} \\
  \mathit{Ateb}(e[x\la v]) & = & \tilde\mu y.\coupe{\mathit{Ateb}(v)}{\tilde\mu x.\coupe{y}{\mathit{Ateb}(e)}} &
  \mbox{with } y \mbox{ fresh} \\
  \mathit{Ateb}(e[\alpha\la e']) & = & \tilde\mu x.\coupe{\mathit{Ateb}(e')\cdot x}{\alpha\lambda.\mathit{Ateb}(e)} &
  \mbox{with } x \mbox{ fresh} \\
\end{array}
$$

It is obvious that for all $t$,
$\mathit{Ateb}(t)$ does not contain substitutions. We must check firstly that
the returned term is typeable, and secondly that it reduces to the original term.

\begin{lemma}
  $$\Gamma\vdash t:A \;\;\Ra\;\; \Gamma\vdash \mathit{Ateb}(t):A$$
\end{lemma}

\begin{proof}
  By induction of the typing derivation of $t$. The only interesting cases are those of
  substitutions.

  \begin{itemize}
    \item We type $c[x\la v]$

      $$
      \infer
        {c[x\la v]:(\Gamma\vdash\Delta)}
        {c:(\Gamma,x:A\vdash\Delta) &
          \infer
            {[x\la v]:(\Gamma,x:A\vdash\Delta)\Ra(\Gamma\vdash\Delta)}
            {\Gamma\vdash v:A|\Delta}}
      $$

      By induction hypothesis, we have $\mathit{Ateb}(c):(\Gamma,x:A\vdash\Delta)$ and
      $\Gamma\vdash \mathit{Ateb}(v):A|\Delta$. We can type $\mathit{Ateb}(c[x\la v])=
      \coupe{\mathit{Ateb}(v)}{\tilde\mu x.\mathit{Ateb}(c)}$ as follows

      $$
      \infer
      {\coupe{\mathit{Ateb}(v)}{\tilde\mu x.\mathit{Ateb}(c)}:(\Gamma\vdash\Delta)}
      {\Gamma\vdash \mathit{Ateb}(v):A|\Delta &
        \infer
        {\Gamma|\tilde\mu x.\mathit{Ateb}(c):A\vdash\Delta}
        {\mathit{Ateb}(c):(\Gamma,x:A\vdash\Delta)}}
      $$
    \item The case $c[\alpha\la e]$ is similar to the previous one by symmetry.

    \item We type $v[x\la v']$

      $$
      \infer
        {\Gamma\vdash v[x\la v']:A|\Delta}
        {\Gamma,x:B\vdash v:A|\Delta &
          \infer
            {[x\la v']:(\Gamma,x:B\vdash\Delta)\Ra(\Gamma\vdash\Delta)}
            {\Gamma\vdash v':B|\Delta}}
      $$

      By induction hypothesis, we have $\Gamma,x:B\vdash \mathit{Ateb}(v):A|\Delta$ and
      $\Gamma\vdash \mathit{Ateb}(v'):B|\Delta$. We can type $\mathit{Ateb}(v[x\la v'])=
      \mu\alpha.\coupe{\lambda x.\mathit{Ateb}(v)}{\mathit{Ateb}(v')\cdot\alpha}$
      as follows

      $$
      \infer
        {\Gamma\vdash\mu\alpha.\coupe{\lambda x.\mathit{Ateb}(v)}{\mathit{Ateb}(v')\cdot\alpha}:A|\Delta}
        {\infer
          {\coupe{\lambda
              x.\mathit{Ateb}(v)}{\mathit{Ateb}(v')\cdot\alpha}:(\Gamma\vdash\Delta,\alpha:A)}
          {\infer
            {\Gamma\vdash\lambda x.\mathit{Ateb}(v):B\ra A|\Delta,\alpha:A}
            {\infer
              {\Gamma,x:B\vdash\mathit{Ateb}(v):A|\Delta,\alpha:A}
              {\Gamma,x:B\vdash \mathit{Ateb}(v):A|\Delta}}&
            \infer
              {\Gamma|\mathit{Ateb}(v')\cdot\alpha:B\ra A\vdash\Delta,\alpha:A}
              {\infer
                {\Gamma\vdash\mathit{Ateb}(v'):B|\Delta,\alpha:A}
                {\Gamma\vdash\mathit{Ateb}(v'):B|\Delta} &
                \Gamma|\alpha:A\vdash\Delta,\alpha:A}}}
      $$
    \item We type $v[\alpha\la e]$

      $$
      \infer
        {\Gamma\vdash v[\alpha\la e]:A|\Delta}
        {\Gamma\vdash v:A|\Delta,\alpha:B &
          \infer
            {[\alpha\la e]:(\Gamma\vdash\Delta,\alpha:B)\Ra(\Gamma\vdash\Delta)}
            {\Gamma\vdash e:B|\Delta}}
      $$

      By induction hypothesis, we have $\Gamma\vdash \mathit{Ateb}(v):A|\Delta,\alpha:B$ and
      $\Gamma\vdash \mathit{Ateb}(e):B|\Delta$. We can type $\mathit{Ateb}(v[\alpha\la e])=
      \mu\beta.\coupe{\mu\alpha.\coupe{\mathit{Ateb}(v)}{\beta}}{\mathit{Ateb}(e)}$
      as follows

      $$
      \infer
        {\Gamma\vdash\mu\beta.\coupe{\mu\alpha.\coupe{\mathit{Ateb}(v)}{\beta}}{\mathit{Ateb}(e)}:A|\Delta}
        {\infer
          {\coupe{\mu\alpha.\coupe{\mathit{Ateb}(v)}{\beta}}{\mathit{Ateb}(e)}:(\Gamma\vdash\Delta,\beta:A)}
          {\infer
            {\Gamma\vdash\mu\alpha.\coupe{\mathit{Ateb}(v)}{\beta}:B|\Delta,\beta:A}
            {\infer
              {\coupe{\mathit{Ateb}(v)}{\beta}:(\Gamma\vdash\Delta,\beta:A,\alpha:B)}
              {\infer
                {\Gamma\vdash\mathit{Ateb}(v):A|\Delta,\beta:A,\alpha:B}
                {\Gamma\vdash \mathit{Ateb}(v):A|\Delta,\alpha:B} &
               \Gamma|\beta:A\vdash\Delta,\beta:A,\alpha:B}} &
           \infer
             {\Gamma\vdash \mathit{Ateb}(e):B|\Delta,\beta:A}
             {\Gamma\vdash \mathit{Ateb}(e):B|\Delta}}}
      $$

    \item The cases for $e[*\la t]$ are similar to the previous ones by symmetry.
    \end{itemize}
\end{proof}

\begin{lemma}
  $$\mathit{Ateb}(t)\ra^*t$$
\end{lemma}

\begin{proof}
  By induction on $t$. The only interesting cases are those of
  substitutions.

  \begin{itemize}
    \item We have $\mathit{Ateb}(c[x\la v])=\coupe{\mathit{Ateb}(v)}{\tilde\mu
        x.\mathit{Ateb}(c)}$ and

      $$
      \coupe{\mathit{Ateb}(v)}{\tilde\mu x.\mathit{Ateb}(c)}\ra_\mu
      \mathit{Ateb}(c)[x\la \mathit{Ateb}(v)]
      $$

      We conclude by induction hypothesis.
    \item The case $c[\alpha\la e]$ is similar to the previous one by symmetry.

    \item We have $\mathit{Ateb}(v[x\la v'])=
      \mu\alpha.\coupe{\lambda x.\mathit{Ateb}(v)}{\mathit{Ateb}(v')\cdot\alpha}$ and

      $$
      \begin{array}{c}
        \mu\alpha.\coupe{\lambda x.\mathit{Ateb}(v)}{\mathit{Ateb}(v')\cdot\alpha} \\
        \da\beta \\
        \mu\alpha.\coupe{\mathit{Ateb}(v')}{\tilde\mu x.\coupe{\mathit{Ateb}(v)}{\alpha}} \\
        \da\tilde\mu \\
        \mu\alpha.(\coupe{\mathit{Ateb}(v)}{\alpha}[x\la \mathit{Ateb}(v')]) \\
        \da c\tau \\
        \mu\alpha.\coupe{\mathit{Ateb}(v)[x\la\mathit{Ateb}(v')]}{\alpha[x\la \mathit{Ateb}(v')]} \\
        \da \alpha\tau2 \\
        \mu\alpha.\coupe{\mathit{Ateb}(v)[x\la \mathit{Ateb}(v')]}{\alpha} \\
        \da sv \\
        \mathit{Ateb}(v)[x\la \mathit{Ateb}(v')]
      \end{array}
      $$

      We conclude by induction hypothesis.

    \item We have $\mathit{Ateb}(v[\alpha\la e])=
      \mu\beta.\coupe{\mu\alpha.\coupe{\mathit{Ateb}(v)}{\beta}}{\mathit{Ateb}(e)}$ and

      $$
      \begin{array}{c}
        \mu\beta.\coupe{\mu\alpha.\coupe{\mathit{Ateb}(v)}{\beta}}{\mathit{Ateb}(e)} \\
        \da\mu \\
        \mu\beta.(\coupe{\mathit{Ateb}(v)}{\beta}[\alpha\la \mathit{Ateb}(e)]) \\
        \da c\tau \\
        \mu\beta.\coupe{\mathit{Ateb}(v)[\alpha\la \mathit{Ateb}(e)]}{\beta[\alpha\la
          \mathit{Ateb}(e)]} \\
        \da \alpha\tau2 \\
        \mu\beta.\coupe{\mathit{Ateb}(v)[\alpha\la \mathit{Ateb}(e)]}{\beta} \\
        \da sv \\
        \mathit{Ateb}(v)[\alpha\la \mathit{Ateb}(e)]
      \end{array}
      $$

      We conclude by induction hypothesis.

    \item The cases for $e[*\la t]$ are similar to the previous ones by symmetry.
    \end{itemize}
\end{proof}

\section{Conclusion}

The technique formalized here gives a new tool to prove strong normalization of
calculi with explicit substitutions. As we have seen, the principle of the proof technique
is simple, and the difficulties arise in the definition of the reverse rewriting rule
that must satisfy precise criteria.

We applied this technique to several calculi, yielding the following results:
\begin{itemize}
\item $\lambda{\tt x}$: there is here no novelty since it is this case that
  originally inspired the technique.
\item $\lambda\upsilon$: we gives here the first strong normalization proof for
  this calculus.
\item $\ls$: this calculus does not enjoy PSN, but we showed that no further objection
  relies to prove strong normalization.
\item $\lsn$: as above.
\item $\lws$: the technique seems to fail due to the presence of labels. Further investigations
  would be necessary to find how this can be fixed..
\item $\lwsn$: the technique can be used, even if this calculus has currently no proof
  of PSN.
\end{itemize}

\bibliography{biblio}

\begin{thebibliography}{10}

\bibitem{AbaCarCurLev91}
M.~Abadi, L.~Cardelli, P.-L. Curien, and J.-J. L\'evy.
\newblock Explicit substitutions.
\newblock {\em Journal of Functional Programming}, 1991.

\bibitem{Bar81}
H.~P. Barendregt.
\newblock {\em The Lambda Calculus : its Syntax and Semantics}.
\newblock 1981.

\bibitem{Les96}
Z.-E.-A. Benaissa, D.~Briaud, P.~Lescanne, and J.~Rouyer-Degli.
\newblock $\lambda\upsilon$, a calculus of explicit substitutions which
  preserves strong normalisation.
\newblock {\em Journal of Functional Programming}, 1996.

\bibitem{Blo97}
R.~Bloo.
\newblock {\em Preservation of Termination for Explicit Substitution}.
\newblock PhD thesis, Eindhoven University, 1997.

\bibitem{BloGeu99}
R.~Bloo and H.~Geuvers.
\newblock Explicit substitution: on the edge of strong normalisation.
\newblock {\em Theoretical Computer Science}, 211:375--395, 1999.

\bibitem{BloRos95}
R.~Bloo and K.H. Rose.
\newblock Preservation of strong normalization in named lambda calculi with
  explicit substitution and garbage collection.
\newblock {\em Computer Science in the Netherlands (CSN)}, 1995.

\bibitem{Chu41}
A.~Church.
\newblock {\em The Calculi of Lambda Conversion}.
\newblock Princeton University Press, 1941.

\bibitem{CosKesPol03}
R.~Di Cosmo, D.~Kesner, and E.~Polonovski.
\newblock Proof nets and explicit substitutions.
\newblock {\em Mathematical Structures in Computer Science}, 13(3):409--450,
  2003.

\bibitem{CurRio91}
P.-L. Curien and A.~R\'{\i}os.
\newblock Un r\'esultat de compl\'etude pour les substitutions explicites.
\newblock {\em Comptes rendus de l'acad\'emie des sciences de Paris}, t. 312,
  S\'erie I:471--476, 1991.

\bibitem{DavGui03}
R.~David and B.~Guillaume.
\newblock Strong normalisation of the typed $\lambda_{ws}$-calculus.
\newblock In {\em Proceedings of CSL'03}, volume 2803 of {\em LNCS}. Springer,
  2003.

\bibitem{Kri90}
J.-L. Krivine.
\newblock {\em Lambda-calcul, types et modèles}.
\newblock Masson, 1990.

\bibitem{Les94}
P.~Lescanne.
\newblock From lambda-sigma to lambda-upsilon: a journey through calculi of
  explicit substitutions.
\newblock In {\em Proceedings of the 21st ACM Symposium on Principles of
  Programming Languages (POPL)}, pages 60--69, 1994.

\bibitem{PolPhD04}
E.~Polonovski.
\newblock {\em Substitutions explicites, logique et normalisation}.
\newblock Th\`ese de doctorat, Université Paris VII, 2004.

\end{thebibliography}

\end{document}